\documentclass{amsart}
\usepackage{amsthm,amsmath,amssymb,amsfonts,hyperref,cleveref,tikz,tikz-cd,enumitem,stmaryrd,xr,quiver,adjustbox,physics}
\usetikzlibrary{positioning,arrows.meta,fit,calc}
\allowdisplaybreaks[1]
\DeclareMathOperator{\Split}{Split}
\DeclareMathOperator{\LR}{LR}
\newcommand{\DD}{\mathbb{D}}
\newcommand{\cat}[1]{\mathbf{#1}}
\newcommand{\sem}[1]{\ensuremath{\llbracket #1 \rrbracket}}
\newcommand{\axref}[1]{\ensuremath{\left(\ref{#1}\right)}}
\newcommand{\id}[1][]{\ensuremath{\mathrm{id}_{#1}}}
\newcommand{\myid}{\ensuremath{\mathrm{id}}}
\newcommand{\prop}[1]{\ensuremath{\langle #1 \rangle}}

\newcommand{\zetak}[1][k]{\ensuremath{\zeta_{#1}}}
\newcommand{\N}{\ensuremath{\mathbb{N}}}
\newcommand{\im}{\ensuremath{\mathrm{im}}}
\newcommand{\new}{\mathsf{new}}
\newcommand{\measure}{\mathsf{measure}}
\newcommand{\apply}{\mathsf{apply}}
\newcommand{\distribli}{\ensuremath{\delta^{\text{l}}}}
\newcommand{\distribri}{\ensuremath{\delta^{\text{r}}}}
\newtheorem{theorem}{Theorem}
\newtheorem{lemma}[theorem]{Lemma}
\newtheorem{proposition}[theorem]{Proposition}
\newtheorem{corollary}[theorem]{Corollary}
\theoremstyle{definition}
\newtheorem{definition}[theorem]{Definition}

\title{Free Quantum Computing}
\author[J. Carette]{Jacques Carette}
\address{McMaster University, Department of Computing and Software, Canada}
\email{carette@mcmaster.ca}
\author[C. Heunen]{Chris Heunen}
\address{University of Edinburgh, School of Informatics, United Kingdom}
\email{chris.heunen@ed.ac.uk}
\author[R. Kaarsgaard]{Robin Kaarsgaard}
\address{University of Southern Denmark, Department of Mathematics and Computer Science, Centre for Quantum Mathematics, Centre for Formal Methods and Future Computing, Denmark}
\email{kaarsgaard@imada.sdu.dk}
\author[N. J. Ross]{\\Neil J. Ross}
\address{Dalhousie University, Department of Mathematics and Statistics, Canada}
\email{neil.jr.ross@dal.ca}
\author[A. Sabry]{Amr Sabry}
\address{Indiana University, Department of Computer Science, United States of America}
\email{sabry@iu.edu}
\thanks{We thank John Baez and J{\o}rgen Ellegaard Andersen for pointing out the topological nature of Axiom~\axref{eq:axiom2}.
Robin Kaarsgaard was supported by Sapere Aude: DFF-Research Leader grant 5251-00024B. 
Neil J. Ross was supported by the Natural Sciences and Engineering Research Council of Canada (NSERC)}

\begin{document}
\begin{abstract}
    Quantum computing improves substantially on known classical algorithms for various important problems, but the nature of the relationship between quantum and classical computing is not yet fully understood.
    This relationship can be clarified by free models, that add to classical computing just enough physical principles to represent quantum computing and no more.
    Here we develop an axiomatisation of quantum computing that replaces the standard continuous postulates with a small number of discrete equations, as well as a free model that replaces the standard linear-algebraic model with a category-theoretical one.
    The axioms and model are based on reversible classical computing, isolate quantum advantage in the ability to take certain well-behaved square roots, and link to various quantum computing hardware platforms.
    This approach allows combinatorial optimisation, including brute force computer search, to optimise quantum computations.
    The free model may be interpreted as a programming language for quantum computers, that has the same expressivity and computational universality as the standard model, but additionally allows automated verification and reasoning.
\end{abstract}

\maketitle

Quantum computing improves substantially on known classical algorithms for certain problems~\cite{google:supremacy} including prime factorisation~\cite{shor:factoring}, boson sampling~\cite{lundetal:bosonsampling}, and Hamiltonian simulation~\cite{abramslloyd:simulation}.
The nature of the relationship between classical and quantum computing is not yet fully understood. 
The literature only identifies several notions that do not explain the difference on their own~\cite{vedral:elusive}, including superposition~\cite{deutsch:universal}, entanglement~\cite{vidal:entanglement}, nonlocality~\cite{jozsalinden:entanglement}, contextuality~\cite{howardetal:contextuality}, and interference~\cite{gottesmanknill}.
The advantage of quantum computing over classical computing is often considered quantitatively: just how much advantage does a specific quantum algorithm have over classical ones~\cite{aaronsonchen:supremacy}? Here we consider it qualitatively: does a model of computation allow algorithms that have advantage over classical ones, or does it not?

The question is difficult partly because quantum computing is often implicitly confined to a \emph{standard model}, that uses complex linear maps~\cite{nielsenchuang}. 
For a meaningful comparison, both classical and quantum computing must be situated within a larger landscape of models.
We use as models bipermutative categories, the most permissive notion of model of computation possible.
Within this framework, we identify a unified \emph{free model} of quantum computing that can express all quantum algorithms with advantage over classical ones.

To explain the free property of a model by way of example, consider combinatorics.
The natural numbers can be completed with negatives in multiple ways: $\mathbb{N}$ embeds into $\mathbb{Z}$ while preserving addition as $n \mapsto n$, but also as $n \mapsto -n$ or $n \mapsto 2n$.
The fact that these are \emph{completions with negatives} means that an additive inverse for the number 1 is adjoined, and entails that all these embeddings are in fact the same up to an unimportant global scalar. 
The fact that it is the \emph{free} completion means that any other completion has to encompass it.
The real numbers also contain $\mathbb{N}$ and allow subtraction, but additionally have many other properties that are superfluous for counting. 
Any addition-preserving embedding $\mathbb{N} \to \mathbb{R}$ factors uniquely through the free model $\mathbb{Z}$ (of natural numbers with negatives).
In this sense, the integers are the smallest model, devoid of extraneous assumptions or constraints, of natural numbers with negatives. This is exactly what being a \emph{free model} captures.
See Figure~\ref{fig:free}.
Generally, free models are unique: they are completely determined by their defining properties (such as having negatives) up to a unique isomorphism (such as an unimportant global scalar).

This article explains the relationship between classical and quantum computing in a similar way. 
Like adjoining $-1$ to the natural numbers, we adjoin a small number of generators with physical significance to reversible classical computing and show that reversible quantum computing arises as the free completion. 
Like the real numbers, any other model of reversible quantum computing, including the standard model of complex unitary matrices, must factor uniquely through this free one.
Conversely, we can canonically identify classical computing within any model of quantum computing by the latter's (bipermutative) structure.
The free model has the same computational universality as the standard model: it can express any quantum algorithm to any desired accuracy. But it has additional virtues.

The main advantage of this free model over the standard model is that it enables automated reasoning about quantum algorithms. For example, deciding whether two quantum circuits represent the same quantum computation is an intractable problem~\cite{janzingetal:qma}, among other reasons because there are continuously many degrees of freedom. But because the free model is entirely discrete and concerns only finitely many equations, the full power of combinatorial optimisations and computer science methods of equational rewriting can be brought to bear~\cite{nametal:optimisation}. As a simple example, we will prove that the Hadamard gate is an involution; this fact is not deep, but the point is that this equality can be derived automatically, as can all others.

A second advantage of the free model is foundational: it does not conflate physical concepts. We thus isolate the difference between quantum and classical computation in the ability to take well-behaved square roots, or in other words, the ability to stop some classical computations half-way. 
The free model is explained from operational first principles, unlike the foundationally unsatisfactory reliance on complex numbers~\cite{delsantogisin:reals}. 
Similarly, it does not need or allow intermediate unphysical matrices that have to be made unitary again at a later stage, circumventing the inefficiency of circuit synthesis that plagues other approaches~\cite{debeaudrapetal:zxhard}.

A third advantage of the free model is practical: it does not resort to unempirical numbers that cannot experimentally be determined with arbitrary precision in finite time, but accounts for accuracy precisely. 
We will not discuss quantitative results in individual models of computation, like the density theorems~\cite{solovaykitaev,rossselinger:synthesis} that show how efficiently a unitary matrix can be approximated by quantum circuits. 
Nevertheless, as an example we will analyse the relationship between accuracy in the free model and efficiency of its implementation of the quantum Fourier transform.

A substantial literature has explored structural and equational approaches to quantum computation, most notably categorical quantum mechanics~\cite{abramskycoecke:cqm,heunenvicary:cqm} and diagrammatic formalisms such as the ZX-calculus~\cite{coeckeduncan:zx}.
These frameworks replace matrices by generators and rewrite rules, enabling symbolic manipulation and circuit optimisation. Our aim is related but conceptually distinct. 
Rather than giving sound and often complete presentations of specific models---typically Hilbert spaces or fragments thereof---we construct the free model generated by a set of axioms with physical connotations. The result is not a calculus for reasoning within a fixed quantum model, but a universal quantum programming language for that axiomatisation.

\section*{Auxiliary qubits}

The quantum computing literature focuses mostly on universality of finite gate sets, rather than explicit comparison to classical computation~\cite{shi:aharonov}.
Here, our focus is on models of quantum computation that reuse the infrastructure of classical computation and are reached by a finite set of axioms.
We fix one set of axioms, containing square roots and a precision parameter $k$, and investigate its free model. 
We show that the precision level can be increased by $n$ at the cost of $n$ auxiliary qubits.
For $k=2$ the free model encompasses Clifford+Toffoli quantum circuits; $k=3$ encompasses Clifford+$T$ quantum circuits; $k=4$ encompasses Clifford+$T$+$\sqrt{T}$ quantum circuits; and higher precision encompasses Clifford cyclotomic quantum circuits~\cite{cyclotomic}.
Other axioms, for example using cube roots, give other free models. See Figure \ref{fig:models}.
In general it makes no sense to compare models for different axioms, but we conjecture that as $k$ tends to infinity, the free model of the axioms using square roots coincides with that using cube roots or roots for any other radix. Our axioms distinguish themselves by having physical connections beyond the mathematical and computational advantages discussed above.

\begin{figure}
  \begin{tikzpicture}[xscale=3,yscale=1.5]
    \node (tl) at (0,1) {Axioms};
    \node (tr) at (1,1) {Free model};
    \node (br) at (1,0) {Any model};
    \draw[->] (tl) to (tr);
    \draw[->] (tl) to (br);
    \draw[->, dashed] (tr) to node[right]{$\exists!$} (br);
  \end{tikzpicture}
  \caption{The defining property of a free model: it realises our axioms of quantum computation, and there is a unique interpretation of the free model in any other model satisfying our axioms. In this sense, the free model contains exactly what is needed to form a model, and nothing more.}
  \label{fig:free}
\end{figure}
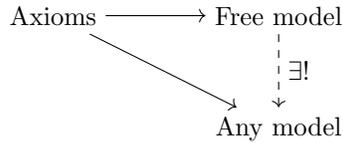

\begin{figure}[t]
{\centering
\resizebox{0.5\textwidth}{!}{
  \begin{tikzpicture}[
    >=Latex,
    node distance=10mm and 16mm,
    model/.style={draw, rounded corners, align=center, inner sep=4pt},
    incl/.style={-Latex, line width=0.8pt},
    factor/.style={-Latex, dashed, line width=0.8pt},
    incomparable/.style={densely dotted, line width=0.6pt}
  ]

  \node[model] (std) {Standard model\\[2pt] $\mathrm{Unitary}(\mathbb{C})$};
  \node[model, below left=22mm and -2mm of std] (k2) {Clifford+Toffoli};
  \node[model, right=10mm of k2] (k3) {Clifford+$T$};
  \node[model, right=10mm of k3] (k4) {Clifford+T+$\sqrt{T}$};
  \node[model, right=10mm of k4] (ccxh) {Clifford+$\sqrt[3]{T}$};
  \node[model, below=22mm of k2] (f2) {$\Pi_{2}$};
  \node[model, below=22mm of k3] (f3) {$\Pi_{3}$};
  \node[model, below=22mm of k4] (f4) {$\Pi_{4}$};
  \node[model, right=27mm of std, yshift=-10mm] (others) {Other models\Large$\dots$};

  \draw[incl] (k2) -- (k3);
  \draw[incl] (k3) -- (k4);
  \draw[incl] (k2) -- (std);
  \draw[incl] (k3) -- (std);
  \draw[incl] (k4) -- (std);
  \draw[incl] (f2) -- (k2);
  \draw[incl] (f3) -- (k3);
  \draw[incl] (f4) -- (k4);
  \draw[incl] (f2) -- (f3);
  \draw[incl] (f3) -- (f4);
  \draw[incl] (ccxh) -- (std);
  \draw[incl] (others) -- (std);

  \draw[incomparable] (ccxh) -- (f4);
  \end{tikzpicture}
}}
\caption{The landscape of models of quantum
  computing and their relations. Solid arrows indicate inclusion
  of one model into another. Incomparable models are connected with dotted lines. 
  For each $k$, $\Pi_k$ is the free model for
  the axioms of bipermutative categories augmented with axioms \axref{eq:axiom1}--\axref{eq:axiom3}. The standard model of unitaries over the complex
  numbers also satisfies these axioms. For each $k$, $\Pi_k$ embeds in
  any model that also satisfies the same axioms. For another example, the model of Clifford+$\sqrt[3]{T}$ circuits is also computationally universal~\cite{nebe_invariants_2001} but incomparable to any $\Pi_k$.}
\label{fig:models}
\end{figure}
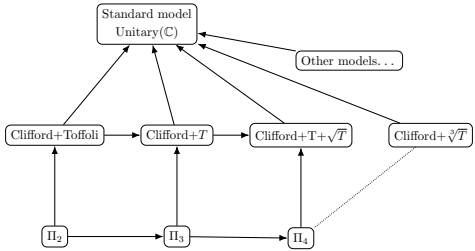

In addition to increased precision, auxiliary qubits also provide the opportunity to extend our free model of reversible quantum computing with quantum measurement. A second free construction allows mid-computation measurements to control the flow of the quantum computation~\cite{heunenkaarsgaard:informationeffects}. We will show that the ability to measure out auxiliary qubits strengthens the completeness of our free model: two irreversible quantum computations are equal if and only if they follow from our axioms, without any need for auxiliary qubits.

\section*{Axioms}

The standard postulates of reversible quantum computing may be replaced by the following; we will discuss (irreversible) quantum measurements later.
The axioms of reversible quantum computing, with precision level $k$, extend reversible classical computing with two generators $\zeta_k$ and $V$ and three relations
\begin{align}
  V^2 & = X \label{eq:axiom1} \\
  V\,S\,V &= S\,V\,S \label{eq:axiom2} \\
  \smash{\zeta_k^{2^k}} & = 1 \label{eq:axiom3} 
\end{align}
where $X$ is a coherent quantum version of the classical NOT gate, and $S$ is defined in terms of $\zeta_k$, as explained in more detail below~\cite{caretteetal:sqrtpi}.
Classical reversible computing corresponds to the free bipermutative category with no further assumptions, which embeds into any other bipermutative category.

These axioms isolate (well-behaved) square roots as a singular source of any difference between quantum computing and classical computing.
We will also consider a weaker form of equivalence $\approx_k$ which identifies computations when they are equal in the presence of auxiliary qubits. This weak equivalence becomes irrelevant when including (irreversible) quantum measurement and state preparation.

\section*{Interpretation}

Before we go into detail, observe the advantages of this axiomatisation.
It consists of a finite number of equations, without continuous variables or logical quantifiers. Therefore they are eminently amenable to automation, and can benefit from combinatorial optimisation, including brute force computer search~\cite{kortevygen:combinatorialoptimization}.

Moreover, all three equations have direct realisations in important physical models of computation, along with links to conceptual principles, notably superposition and interference.

Axiom~\axref{eq:axiom1} can be read as saying that the $X$ gate has a square root.
This is enacted in trapped-ion quantum computing by pulsing the laser for half the duration or intensity needed to enact the $X$ gate~\cite{winelandetal:trappedions}. 
Similarly, $\sqrt{X}$ is a native gate in superconducting quantum computers~\cite{google:supremacy}.
The $X$ gate having a square root induces superposition, but here the ability to halt some computations halfway is primary, and the existence of some superpositions is merely derived. Constructing all superpositions is impossible in a model that extends classical reversible computing with only finitely many axioms.

Axiom~\axref{eq:axiom2} is the defining equation for the 3-strand braid group, linking this axiomatisation to anyonic quantum computing~\cite{simon:topologicalquantum}. 
See \Cref{fig:braid}.
\begin{figure}
  \tikzset{halo/.style={
           preaction={draw,white,line width=4pt,-},
           preaction={draw,white,ultra thick, shorten >=-2.5\pgflinewidth}}}
  $\begin{aligned}\begin{tikzpicture}[scale=.6]
      \draw (1,0) to[out=90,in=-90] (0,1);
      \draw[halo] (0,0) to[out=90,in=-90] (1,1);
      \draw (2,0) to (2,1) to[out=90,in=-90] (1,2);
      \draw (2,1) to[out=90,in=-90] (1,2);
      \draw (1,2) to[out=90,in=-90] (0,3);
      \draw[halo] (1,1) to[out=90,in=-90] (2,2) to (2,3);
      \draw[halo] (0,1) to (0,2) to[out=90,in=-90] (1,3);
    \end{tikzpicture}\end{aligned} 
    \qquad = \qquad 
    \begin{aligned}\begin{tikzpicture}[scale=.6]
      \draw (2,0) to[out=90,in=-90] (1,1) to[out=90,in=-90] (0,2) to (0,3);
      \draw[halo] (1,0) to[out=90,in=-90] (2,1) to (2,2);
      \draw[halo] (0,0) to (0,1) to[out=90,in=-90] (1,2);
      \draw (2,2) to[out=90,in=-90] (1,3);
      \draw[halo] (1,2) to[out=90,in=-90] (2,3);
  \end{tikzpicture}\end{aligned}$

  \caption{An illustration of Axiom~\axref{eq:axiom2}, interpreting $V$ as swapping the leftmost two wires, and $S$ as swapping the rightmost two wires.}
  \label{fig:braid}
\end{figure}
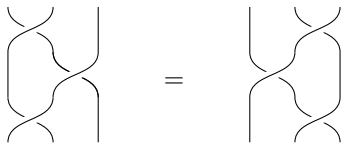

Axiom~\axref{eq:axiom3} concerns constructivity. 
General complex numbers in principle take forever to determine with infinite precision.
The free model instead uses numbers that one can empirically construct or measure exactly.
Positing a nontrivial phase, in combination with the superposition afforded by $V$, induces interference.
We choose a primitive $2^k$th root of unity $\zeta_k = e^{2\pi i/2^k}$ as our primitive nontrivial phase.
The scalars in the free model 
(as defined in ``Bipermutative categories: formalising the main results'' below) 
are therefore the cyclic group of order $2^k$, consisting of powers of $\zeta_k$. 
One can clearly approximate any complex number on the unit circle with such constructive scalars by choosing $k$ sufficiently large, additionally making these scalars useful in parametrised rotation gates, photonic quantum computing, and variational quantum algorithms. 
Notice that the free model has no need for addition of scalars, but the standard model of unitary matrices uses many more scalars.

\section*{Reversible classical computing}

To understand the axiomatisation and explain what we mean by a \emph{model} in more detail, we first discuss \emph{reversible classical computing}. 
In this paradigm, computations are deterministic when run either forwards or backwards~\cite{gluckyokoyama:reversible}.
It relates to physics, as experiments have verified that irreversibly changing information stored in a computer dissipates heat~\cite{berutetal:landauer}, in accordance with  Landauer's principle~\cite{landauer,landauer:informationisphysical}.

The standard model of classical reversible computing is given by circuits of reversible Boolean gates.
For example, the gate $CCX$, also known as the Toffoli gate, which transforms three input bits $(x,y,z) \in \{0,1\}^3$ into $(x,y, xy+z \;\textrm{mod}\;2)$, is universal in the sense that any Boolean function can be implemented by a circuit of CCX gates.

A general model for classical reversible computing is given by the notion of a \emph{bipermutative category}, which has objects and processes that can be combined using combinators $\oplus$ and $\otimes$~\cite{choudhurykawowskisbary:symmetries}. 
This is the bare minimum that any model of computation needs to represent: the ability to compose instructions sequentially (using $\circ$), the ability to consider data in parallel (using $\otimes$), and the ability to use one such concurrent piece of data to influence another (using $\oplus$).
In the standard model, the unit object $1$ for $\otimes$ is analogous to a singleton set, and $2=1 \oplus 1$ models the bit. The swap process $X \colon 1 \oplus 1 \to 1 \oplus 1$ models a NOT gate, and $CCX$ is typical of modelling controlled operations as $\mathrm{id} \oplus X \colon 2 \oplus 2\oplus 2\oplus 2 = 2 \otimes 2 \otimes 2 \to 2 \otimes 2 \otimes 2$.

The free model of classical reversible computing consists of permutations of finite sets, or equivalently, permutation matrices.
Because every permutation of a finite set is a composition of transpositions, the NOT gate $X$  generates all of these.
The free bipermutative category is thus identified with the free model of classical reversible computing: the $X$ and $CCX$ gates are already modeled by the axioms of bipermutative categories.

\section*{Reversible quantum computing}

Similar combinators $\oplus$ and $\otimes$ extend quantum computing from a single qubit to multiple controlled qubits.
We can now understand Axioms~\axref{eq:axiom1}--\axref{eq:axiom3} in detail. Start with the free bipermutative category, \textit{i.e.} the free model of classical reversible computing, and choose an integer $k \ge 2$. Add generators $\zeta_k \colon 1 \to 1$ and $V \colon 2 \to 2$, and consider the free bipermutative category satisfying Axioms~\axref{eq:axiom1}--\axref{eq:axiom3}, where $S=\mathrm{id} \oplus \zeta_k^{2^{k-2}}$.
\Cref{thm:freemodel} below shows that this free model exists. 
Notice that if we only used $\otimes$, adding $\zeta_k$ merely yields physically irrelevant global phases; it is the presence of $\oplus$ that yields relative phases with nontrivial computational consequences.
The combinator $\oplus$ gives the ability to represent qudits for all dimensions rather than only qubits, which harmonises with the need for quantum algorithms such as Shor's prime factoring to allocate storage for values to the nearest power of 2, some of which will not be used in the actual computation. 

The standard model of unitary matrices also satisfies Axioms~\axref{eq:axiom1}--\axref{eq:axiom3} with
\[
  \zeta_k = e^{2\pi i/2^k}\text,
  \qquad 
  V = \frac{1}{2}\begin{pmatrix} 1+i & 1-i \\ 1-i & 1+i  \end{pmatrix} \text,
  \qquad 
  S = \begin{pmatrix} 1 & 0 \\ 0 & i \end{pmatrix}.
\]
Hence there is an interpretation of the free model in the standard matrix model (see \Cref{fig:free}). \Cref{thm:standardmodel} below shows that this interpretation is in fact an isomorphism when $k=2$ to a subset of the standard matrix model. Free models are unique up to unique isomorphism: given two free models, setting one to be the ``Free model'' and instantiating ``Any model'' in Figure~\ref{fig:free} to be the other and vice versa gives two maps that are each other's inverse. We conclude that the free model (which has $k=2$) is formed by unitary matrices with entries from the smallest ring containing $\zeta_2$ and $\frac{1}{2}$. Vice versa, for each $k$, we may think of the free model as a programming language with which to construct reversible quantum programs.

The standard matrix interpretation answers two computer science questions~\cite{rewriting} about this quantum programming language.
First, the word problem is decidable: given two expressions in the programming language, there is an algorithm that can decide whether they compute the same function, by comparing the matrices they induce.
Second, normalisation by evaluation~\cite{151645} is possible: one can optimise a quantum program by synthesising a normal form from the matrix it induces. This is akin to quantum circuit optimisation.

Finally, by foregoing matrices and working with the free model no reasoning power is lost. It follows from \Cref{thm:standardmodel} below that two quantum programs denote the same matrix if and only if they are provably equal under substitutions using Axioms~\axref{eq:axiom1}--\axref{eq:axiom3}. It thus becomes possible to automate verification of quantum programs, and to use automated proof assistants~\cite{norell:agda} to construct, transform, and reason about quantum computations.

\section*{Precision and conditional circuits with auxiliary qubits}

\Cref{thm:density} establishes that for every $k \geq 2$, the
  free model is topologically dense in the standard model of unitary
  matrices over the complex numbers. In other words, any target
  unitary can be approximated to arbitrary accuracy by a program in
  the free model, justifying its status as a model of (reversible)
  quantum computation.

\Cref{thm:precision} further shows that the precision
  parameter~$k$ does not affect expressivity once a clean auxiliary
  qubit is available: $\Pi_{k+1}$ can be simulated exactly by $\Pi_k$
  using a single auxiliary qubit. Thus $k$ only influences the size and
  convenience of terms, not which unitaries are definable. In
  particular, $k=2$ already suffices for universal
  approximation. Larger~$k$ can produce much shorter $\Pi_k$ terms for
  the same target unitary, and can make it easier to reason about
  programs and prove algebraic properties within the syntax, since
  fewer identities must be expressed through long compiled
  sequences. Indeed, certain highly structured unitaries benefit
  dramatically from having more primitive phases. A notable example is
  the quantum Fourier transform (QFT), whose phase angles are exact
  powers of $1/2$, making it an ideal case study for the impact of $k$
  on synthesis cost.

Each $n$-qubit QFT contains $O(n^2)$ controlled-phase gates
  with angles $2\pi/2^d$ for $2\le d\le
  n$~\cite{nielsenchuang}. Because $\Pi_k$ includes all phases with
  denominator $2^k$ as primitives, only smaller-angle rotations
  require approximation, at a cost of $O(\log_2(1/\varepsilon))$ per
  rotation, where $\varepsilon$ is the accuracy of the
  approximation~\cite{kitaevshenvyalyi,rossselinger:synthesis}. Increasing
  $k$ reduces the number of rotations to approximate until $k$ reaches
  $n$, after which all rotations are native and the size plateaus at
  $O(n^2)$. This sharp cutoff is a consequence of the QFT’s
  structured angles; for generic unitaries, which lack such alignment,
  increasing $k$ can only improve constant factors and does not remove
  the $O(\log(1/\varepsilon))$ synthesis overhead.

Beyond reducing the cost of phase synthesis, the $\Pi_k$
  formalism also yields smaller controlled circuits. Recent
  compilation techniques show how to decompose multi-controlled
  operations linearly without helper qubits, by exploiting
  higher-order roots and reusing idle qubits as temporary
  workspace~\cite{10.1145/3656436}. This means that increasing~$k$ not
  only shortens the phase-synthesis portion of a program, but also
  reduces the overhead of conditionals and controlled subroutines,
  further amplifying the practical benefits of higher precision.

\section*{Measurement}

Finally, we discuss a free replacement of the standard postulate about measurement, by incorporating (mid-computation) quantum measurement, initialisation of auxiliary qubits, and classical control into the above account of reversible quantum computing. 
Measurement necessitates working with mixed rather than pure states. 
This means that the objects of the free model can no longer be natural numbers modelling dimensions. 
Instead, they will be lists of natural numbers, modelling dimensions branched on measurement outcomes, enabling classical control.

An irreversible quantum computation from $m$ to $n$ qubits can be simulated perfectly by a reversible quantum computation from $m+n$ to $m+n$ qubits using quantum information effects~\cite{heunenkaarsgaard:informationeffects}. 
This is another free construction that introduces the ability to hide information in a way that respects both sequential composition ($\circ$) and parallel composition ($\otimes$) according to the einselection interpretation of decoherence~\cite{zurek:decoherence}. It may be thought of as a further programming language, in which auxiliary qubits, classical control, and measurements can be used at will anywhere in the quantum program, that gets compiled into a quantum program with a single measurement at the end.
\Cref{thm:measurement} below states that two such programs including measurement denote the same quantum channel if and only if they are provably equal under substitutions using a finite set of equations---see Supporting Information for details. The physical irrelevance of global phase arises in this construction from the uniformity of information hiding, squashing all global phases into a single point.

Similar to how mixed states of $n$ qubits correspond to quantum channels $\mathbb{C} \to (\mathbb{C}^2)^{\otimes n}$, states of some object $A$ in the free model correspond to morphisms $1 \to A$. An input state $\rho \colon 1 \to A$ evolves along an operation $f \colon A \to B$ by composition to give output state $f \circ \rho \colon 1 \to B$. Probabilities are thus a derived concept in the free model, with mixed states given by the equivalence class of their purifications.

\section*{Bipermutative categories: formalising the main results}

Let us now sketch the technical development, leaving full detail to the Supporting Information. 
A \emph{category} consists of objects and morphisms between those objects, such that morphisms $f \colon A \to B$ and $g \colon B \to C$ can be composed to $g \circ f \colon A \to C$ (typically shortened to $gf$) in an associative way, and there are identity morphisms $\mathrm{id} \colon A \to A$ on each object~\cite{maclane:categories}.
Categories can fully characterise quantum theory without referring to the standard model~\cite{heunenkornell:hilbert}.
A \emph{strict symmetric monoidal category} additionally has objects $A \otimes B$ for any two objects $A$ and $B$, and morphisms $f \otimes g \colon A \otimes B \to C \otimes D$ for any morphisms $f \colon A \to C$ and $g \colon B \to D$, a unit object $1$ with $1 \otimes A=A$, and a swap morphism $\sigma \colon A \otimes B \to B \otimes A$, such that $(\mathrm{id} \otimes g) \circ (f \otimes \mathrm{id}) = (f \otimes \mathrm{id}) \circ (\mathrm{id} \otimes g)$ and $\sigma \circ \sigma = \mathrm{id}$~\cite{heunenvicary:cqm}.
Morphisms $s \colon 1 \to 1$ are called \emph{scalars}, any morphism $f \colon A \to B$ can be scaled by a scalar $s$ to form $s \cdot f \colon A \to B$, and this scalar multiplication commutes with all other operations~\cite{heunenvicary:cqm}.
A \emph{bipermutative category} is a category that is strict symmetric monoidal in two ways, $(\otimes,1)$ and $(\oplus,0)$, such that $(A \oplus B) \otimes C=(A \otimes C) \oplus (A \otimes B)$ and three natural equations hold~\cite{may:spectra}.
A category is free with some property when it satisfies the condition of \Cref{fig:free}~\cite{maclane:categories}.

\begin{theorem}\label{thm:freemodel}
  There exists a free bipermutative category $\Pi_k$ with generators $\zeta_k \colon 1 \to 1$ and $V \colon 1 \oplus 1 \to 1 \oplus 1$ satisfying Axioms~\axref{eq:axiom1}--\axref{eq:axiom3}.
\end{theorem}

Recall that, like any free model, $\Pi_k$ is unique up to unique isomorphism.
For example, $\Pi_0$ is the category whose objects are natural numbers, whose morphisms are permutations of $\{1,\ldots,n\}$, and which only has one trivial scalar.
Another example of a bipermutative category is $\cat{Unitary}(R)$, where objects are natural numbers, morphisms are unitary matrices over the involutive ring $R$, and scalars are elements of $R$.
For example, $R$ could be the ring of dyadic rational numbers $\DD = \mathbb{Z}[\frac{1}{2}] = \{ 2^{-n}z \mid n \in \mathbb{N}, z \in \mathbb{Z} \}$, with a primitive $2^k$th root of unity $\zeta_k$ adjoined. 

\begin{theorem}\label{thm:standardmodel}
  There is an isomorphism $\Pi_2 \simeq \cat{Unitary}(\DD[\zeta_2])$ of bipermutative categories.
\end{theorem}
\begin{proof}[Proof sketch]
  Freeness of $\Pi_2$ gives a functor from $\Pi_2$ to $\cat{Unitary}(\DD[\zeta_2])$ that is bijective on objects. 
  That it is bijective on morphisms follows because $\Pi_2$ satisfies the equations generating $\cat{Unitary}(\DD[\zeta_2])$~\cite{caretteetal:sqrtpi}.
\end{proof}

Just as a quantum circuit is merely a formal diagram, morphisms in $\Pi_k$ are the source code of a quantum program. To physically perform the computation, a meaning has to be assigned to it, just like how quantum circuits are interpreted as unitary matrices. This assignment of meaning preserves the structure of computations.

\begin{theorem}\label{thm:density}
  There is an inclusion $\sem{-}$ from $\Pi_k$ into $\cat{Unitary}(\mathbb{C})$ for any precision level $k$, and for $k \geq 2$ its image is dense.
\end{theorem}
\begin{proof}[Proof sketch]
  Already at precision $k=2$, the category $\Pi_2$ contains the $CCX$ gate as well as all Clifford gates (up to a global phase). This suffices to approximate arbitrary complex unitary matrices without auxiliary qubits~\cite{amyglaudellross:numbertheoretic}.
\end{proof}

Similarly, a translation between bipermutative categories is a function that preserves $\circ$ and $\oplus$.

\begin{theorem}\label{thm:precision}
  There is a translation $\Phi$ from $\Pi_{k+1}$ to $\Pi_k$ for $k \geq 2$ such that 
  $f$ and $g$ in $\Pi_{k+1}$ satisfy $\sem{f}=\sem{g}$ if and only if $\sem{\Phi(f)}=\sem{\Phi(g)}$ in $\Pi_{k}$.
\end{theorem}
\begin{proof}[Proof sketch]
  This construction generalises~\cite{shi:aharonov} by emulating a square root of a given scalar using one auxiliary qubit, and noticing that $\Pi_{k+1}$ satisfies this.
\end{proof}

The Supporting Information details a free construction that adds initialisation morphisms, decoherence morphisms, and classical control to a bipermutative category $\Pi_k$ to turn it into a strict monoidal category $\Split(\LR(\Pi_k))$~\cite{heunenkaarsgaard:informationeffects}, and an equivalence $\approx_k$ that identifies computations when they are
equal in the presence of auxiliary qubits.

\begin{theorem}\label{thm:measurement}
  There is an inclusion $\sem{-}$ from $\Split(\LR(\Pi_k))$ to the category of completely positive linear maps between finite-dimensional C*-algebras for any precision level $k$, and $\sem{f}=\sem{g}$ if and only if $f=g$ in $\Split(\LR(\Pi_k))$.
  Computations $f$ and $g$ in $\Pi_k$ are equal in $\Split(\LR(\Pi_k))$ if and
  only if $f \approx_k g$ up to a global phase.
\end{theorem}
\begin{proof}[Proof sketch]
  The category of quantum channels has a known presentation~\cite{staton:quantum}, and it suffices to verify that $\Split(\LR(\Pi_k))$ satisfies its equations.
\end{proof}

\section*{Reasoning}

Reasoning in $\Pi_k$ is exact, symbolic, and complete. 
  Even for a seemingly trivial identity such as $HH = \mathrm{id}$, standard
  approaches face difficulties.  One option is to work
  with matrices whose entries involve algebraic numbers, making exact proofs
  extremely delicate and number-theoretic in nature. While there are
  decision procedures for these, their complexity is unwieldy. The fundamental
  issue is that the theory of algebraic numbers cannot be finitely axiomatised.
  In contrast, $H$ and all other gates are represented in $\Pi_k$ as
  combinatorial objects governed by a small, finite set of
  equations. These equations are both sound (every derivable equality
  holds in all models satisfying the axioms) and complete
  (every valid equality in such models is derivable). As a
  result, all reasoning, from the simplest identities to the most
  involved optimisations, proceeds without approximation, without
  special-purpose constants, and without leaving the symbolic
  framework.

For a simple example, let us prove the identity $HH = \mathrm{id}$ entirely within $\Pi_k$. 
In $\Pi_k$ with $k \geq 3$, define 
\[
  \omega \;:=\; \zeta_k^{2^{k-3}}\text,
  \qquad
  S \;:=\; \mathrm{id} \oplus \omega^2\text,
  \qquad 
  H \;:=\; \omega^{-1} \cdot (V S V)\text.
\]
\noindent
Starting from $H^2$:
\begin{align*}
H^2
&= \big(\omega^{-1} \cdot V S V\big) \;\big(\omega^{-1} \cdot V S V\big) \\
&= \omega^{-2} \cdot V S V V S V &&\text{(centrality of scalars)}\\
&= \omega^{-2} \cdot V \,(S X S)\, V &&\text{(by Axiom~\axref{eq:axiom1}).}
\end{align*}
Unfolding the definition $S = \mathrm{id} \oplus \omega^2$ and using the fact that the $\oplus$-swap respects $\circ$ gives:
\begin{align*}
S X S 
&= (\mathrm{id} \oplus \omega^2) \;X\; (\mathrm{id} \oplus \omega^2) \\
&= (\mathrm{id} \oplus \omega^2) \;(\omega^2 \oplus \mathrm{id})\; X
  && \\ 
&= (\omega^2 \oplus \omega^2) \;X \\
&= \omega^2 \cdot X
   &&\text{(distributivity of scalars)}.
\end{align*}
Substituting back:
\begin{align*}
H^2
&= \omega^{-2} \cdot V \,(\omega^2 \cdot X)\, V \\
&= V X V &&\text{(centrality of scalars)}\\
&= V^4 &&\text{(by Axiom~\axref{eq:axiom1})} \\
&= X^2 \\
&= \mathrm{id} 
   &&\text{(since $X^2 = \sigma_\oplus^2 = \mathrm{id}$)}.
\end{align*}
Thus, $HH = \mathrm{id}$ is derivable in $\Pi_k$ for all $k \ge
  3$ using only the three axioms~\axref{eq:axiom1}--\axref{eq:axiom3} and the axioms of bipermutative categories (which include all algebraic properties of scalars used, such as centrality and distributivity).

\section*{Conclusion}

We have provided an equational axiomatisation of quantum computing and an accompanying free model, that does not use complex numbers and is entirely discrete and combinatorial in nature. The induced quantum programming language suffices to express any quantum computation (as stated in \Cref{thm:density}), and does not lose any universality or reasoning power with respect to the standard model. Complex linear algebra is not required, and in this universal model of quantum computing the full power of combinatorial optimisation can be brought to bear on the problem of manipulating quantum programs in an automated way for optimisation, specification, and verification of correctness.
Whether two Boolean circuits are equivalent is an intractable problem in the worst case~\cite{cooklevin}, but cases arising in practice can efficiently be dealt with heuristically~\cite{smt}. Whether two quantum circuits are equivalent is an even more intractable problem~\cite{janzingetal:qma}, as it additionally requires comparing whether two real numbers are equal. Because it is efficient to decide whether two elements of $\DD[\zeta_2]$ are equal, our free model can be dealt with efficiently with heuristics. 
Furthermore, the axiomatisation is foundationally satisfactory, building on classical reversible computing, linking to various quantum computing hardware platforms, and isolating quantum advantage in the ability to take well-behaved square roots.

The literature mostly works with unitary groups $U_{2^n}(\mathbb{C})$ that are parametrised by the number $n$ of qubits, but otherwise do not observe any relationships between individual groups in the family. Working with bipermutative categories instead means that there are no bounds on the completeness of the free model of \Cref{thm:freemodel} with respect to circuit size or gate count. More precisely, instead of having sets of equations that grow with, and depend on, the number of qubits or number of gates~\cite{xuetal:quartz}, the set of axioms \axref{eq:axiom1}--\axref{eq:axiom3} is complete for all circuits expressible in $\Pi_k$ at once.

\bibliographystyle{plain}
\bibliography{bibliography}

\appendix
\section{Supporting Information}

This Supporting Information contains a proof that the free model of quantum computing described in the article is sound and complete (\Cref{completeness}), and a proof of the claims about the asymptotic scaling of term size as precision varies.
The former proof proceeds in three steps: first we reduce the situation to sums only, obviating the need to consider products; then we prove soundness and completeness in the reversible setting up to auxiliary qubits; and finally we establish soundness and completeness of the full model including measurement. 
We assume familiarity with categorical methods as used in quantum computing~\cite{heunenvicary:cqm} thoughout.

First recall the definition of bipermutative categories~\cite{may:spectra,laplaza:distributivity}.

\begin{definition}\label{bipermutative} 
  A \emph{bipermutative category} is a category $\cat{C}$ with two strict symmetric monoidal structures  $(\otimes,I)$ and $(\oplus,O)$ where
  \begin{itemize}
  \item for all objects $A,B$ and all morphisms $f \colon A \to B$:
    \begin{align*}
      O \otimes A & = O = A \otimes O\\ 
      \id[O] \otimes f & = \id[O] = f \otimes \id[O]
    \end{align*}

   \item for all objects $A,B,C$ and morphisms $f,g,h$:
     \begin{align*}
       (A \oplus B) \otimes C & = (A \otimes C) \oplus (B \otimes C)\\ 
       (f \oplus g) \otimes h & = (f \otimes h) \oplus (g \otimes h)
     \end{align*} 
     and the following diagram commutes:
     \[\begin{tikzcd}
       {(A \oplus B) \otimes C} & {(A \otimes C) \oplus (B \otimes C)} \\
       {(B \oplus A) \otimes C} & {(B \otimes C) \oplus (A \otimes C)}
       \arrow["{=}", from=1-1, to=1-2]
       \arrow["{\sigma_\oplus \otimes \id}"', from=1-1, to=2-1]
       \arrow["{\sigma_\oplus}", from=1-2, to=2-2]
       \arrow["{=}", from=2-1, to=2-2]
     \end{tikzcd}\]

    \item for all objects $A,B,C,D$ the diagram
    \[\begin{tikzcd}
      {(A \oplus B) \otimes (C \oplus D)} & {(A \otimes (C \oplus D)) \oplus (B \otimes (C \oplus D))} \\
      {((A \oplus B) \otimes C) \oplus ((A \oplus B) \otimes D)} & {(A \otimes C) \oplus (A \otimes D) \oplus (B \otimes C) \oplus (B \otimes D)} \\
      & {(A \otimes C) \oplus (B \otimes C) \oplus (A \otimes D) \oplus (B \otimes D)}
      \arrow["{=}", from=1-1, to=1-2]
      \arrow["{\delta_L}"', from=1-1, to=2-1]
      \arrow["{\delta_L \oplus \delta_L}", from=1-2, to=2-2]
      \arrow["{=}"', from=2-1, to=3-2]
      \arrow["{\id \oplus \sigma_\oplus \oplus \id}", from=2-2, to=3-2]
    \end{tikzcd}\]
    commutes, where $\delta_L$ is the composite
      \small\begin{equation*}
        A \otimes (B \oplus C) \xrightarrow{\sigma_\otimes} (B \oplus
        C) \otimes A \xrightarrow{=} (B \otimes A) \oplus (C \otimes
        A) \xrightarrow{\sigma_\otimes \oplus \sigma_\otimes} (A
        \otimes B) \oplus (A \otimes C)\text.
      \end{equation*}\normalsize
  \end{itemize}

  A \emph{strict bipermutative functor} is a functor that is strict symmetric monoidal for both $(\otimes,I)$ and $(\oplus,O)$.

  A bipermutative category is \emph{semisimple} when its objects are
  natural numbers, and its monoidal structures are $n \otimes m = nm$ (with unit $1$) and $n \oplus m = n+m$ (with unit $0$).
\end{definition}

All our free models will be bipermutative categories, but some other models are weaker, such as the standard model of finite-dimensional C*-algebras and quantum channels. They are \emph{rig categories}, where $(A \oplus B) \otimes C$ may not necessarily equal $(A \otimes C) \oplus (B \otimes C)$, but merely be isomorphic to it.
No matter: any rig category is rig equivalent to a bipermutative category~\cite{may:spectra,laplaza:distributivity}.

In any monoidal category $(\cat{C}, \otimes, I)$ one can multiply any
morphism $f \colon A \to B$ with a \emph{scalar} $s \colon I \to I$ as
$s \cdot f \colon A \to B$ as the following composite.

\begin{equation*}
  A \cong I \otimes A \xrightarrow{s \otimes f} I \otimes B \cong B
\end{equation*}

\section*{Tensor products}\label{tensorproducts}

Notice that the axioms~\eqref{eq:axiom1}--\eqref{eq:axiom3} only involve $\circ$ and $\oplus$, not $\otimes$. Moreover, the objects of $\Pi_k$ are simply natural numbers. We will show that in every such category $\otimes$ can be reconstructed from $\oplus$. To be precise, we use \emph{PROP}s (\emph{PROduct and Permutation categories})~\cite{maclane:prop}: 
strict symmetric monoidal categories $(\cat{C},\oplus,0)$ where objects are natural numbers and the sum of objects is addition of natural numbers. 
A \emph{morphism of PROPs} is a strict symmetric monoidal functor between PROPs that is the identity on objects.

For natural numbers $n$ and $m$, define a product $n \cdot m$ on objects in a PROP by the $m$-fold sum of $n$ with itself, and a \emph{Kronecker product} $f \otimes g$ of morphisms $f \colon m \to m'$ and $g \colon n \to n'$ by
\[
  mn \xrightarrow{n \cdot f} nm' \xrightarrow{\pi_{n,m'}} m'n \xrightarrow{m' \cdot g} m'n' \xrightarrow{\pi_{m',n'}} n'm'
\]
where $n \cdot f = \overbrace{f \oplus \cdots \oplus f}^{n \text{ times}}$ and $\pi_{m,n}$ is the reindexing $mn \to nm$ that sends the lexicographic order on the ordinal number $m \times n$ to that on $n \times m$. 
Explicitly~\cite{catetwigg:transposition}:
\[
  \pi_{m,n}(x) = \begin{cases}
    nx \bmod mn-1 & \text{ if } x \neq mn-1\text, \\
    mn-1 & \text{ if } x=mn-1\text.
  \end{cases}
\]
In general, there is no guarantee that $(f \otimes \id) \circ (\id \otimes g) = (\id \otimes g) \circ (f \otimes \id)$.
So a PROP need not induce a semisimple bipermutative category this way, but any semisimple bipermutative category \emph{is} a PROP satisfying extra equations. 
To establish this claim, we will show that every morphism, even one ostensibly defined in terms of products, equals a morphism defined only in terms of sums. For objects this already holds by construction. For morphisms:
\begin{align*}
  f \otimes g
  & = (f \otimes \id[n]) \circ (\id[m] \otimes g) \\
  & = \sigma_\otimes \circ (\id[n] \otimes f) \circ \sigma_\otimes \circ
  (\id[m] \otimes g) \\
  & = \sigma_\otimes \circ (\underbrace{f \oplus \cdots \oplus
  f}_{n~\text{times}}) \circ \sigma_\otimes \circ (\underbrace{g
  \oplus \cdots \oplus g}_{m~\text{times}})\text.
\end{align*}
By the three coherence morphisms in \Cref{bipermutative}, it thus suffices to show that $\sigma_\otimes$ can always be expressed as a composition of sums of $\sigma_\oplus$; because $\delta_L$ is already derived from it by sums, and $\sigma_\oplus$ is already defined in terms of the sum.

\begin{lemma}\label{symmetries}
  In a semisimple bipermutative category, the multiplicative symmetry 
  \[
    \sigma_\otimes \colon (A \oplus B) \otimes (C \oplus D) \to (C \oplus D) \otimes (A \oplus B)
  \]
  can be expressed in terms of the symmetries $\sigma_\otimes$ on the strictly smaller objects $A \otimes (C \oplus D)$, $B \otimes (C \oplus D)$, $C \otimes (A \oplus B)$, $D \otimes (A \oplus B)$, $C \otimes A$, $D \otimes A$, $C \otimes B$, $D \otimes B$, and the additive symmetry $\sigma_\oplus$.
\end{lemma}
\begin{proof}
  For the base case, observe that the symmetries $\sigma_\otimes \colon 1 \otimes n \to n \otimes 1$ and $\sigma_\otimes \colon n \otimes 1 \to 1 \otimes n$ are identity morphisms for any $n$.
  For $m>1$, observe that the following diagram commutes:
  \begin{center}
  \adjustbox{scale=0.66}{
  \begin{tikzcd}
    {(A \otimes (C \oplus D)) \oplus (B \otimes (C \oplus D)) } & {((C \oplus D) \otimes A) \oplus ((C \oplus D) \otimes B) } & {(C  \otimes A) \oplus (D \otimes A) \oplus (C \otimes B) \oplus (D  \otimes B) } \\
    {(A \oplus B) \otimes (C \oplus D)} \\
    && {(A \otimes C) \oplus (A \otimes D) \oplus (B \otimes C) \oplus (B \otimes D)} \\
    {(C \oplus D) \otimes (A \oplus B)} \\
    {(C \otimes (A \oplus B)) \oplus (D  \otimes (A \oplus B))} & {((A \oplus B) \otimes C) \oplus ((A \oplus B) \otimes D)} & {(A \otimes C) \oplus (B \otimes C) \oplus (A \otimes D) \oplus (B \otimes D)}
    \arrow["{\sigma_\otimes \oplus \sigma_\otimes}", from=1-1, to=1-2]
    \arrow[""{name=0, anchor=center, inner sep=0}, "{\delta_L \oplus \delta_L}", from=1-1, to=3-3]
    \arrow["{=}", from=1-2, to=1-3]
    \arrow["{\sigma_\otimes \oplus \sigma_\otimes \oplus \sigma_\otimes \oplus \sigma_\otimes}", from=1-3, to=3-3]
    \arrow["{=}", from=2-1, to=1-1]
    \arrow["{\sigma_\otimes}"', from=2-1, to=4-1]
    \arrow[""{name=1, anchor=center, inner sep=0}, "{\delta_L}"', from=2-1, to=5-2]
    \arrow[""{name=2, anchor=center, inner sep=0}, "{=}"', from=4-1, to=5-1]
    \arrow["{\sigma_\otimes \oplus \sigma_\otimes}"', from=5-1, to=5-2]
    \arrow["{=}"', from=5-2, to=5-3]
    \arrow["{\id \oplus \sigma_\oplus \oplus \id}"', from=5-3, to=3-3]
    \arrow["{(ii)}"{description}, draw=none, from=0, to=1-3]
    \arrow["{(i)}", draw=none, from=0, to=1]
    \arrow["{(ii)}", draw=none, from=1, to=2]
  \end{tikzcd}}
  \end{center}
  Indeed, $(i)$ commutes by coherence condition \cite[(IX)]{laplaza:distributivity} and $(ii)$ by definition of $\delta_L$. 
  Because $\sigma_\otimes \colon B \otimes A \to A \otimes B$ is the inverse of $\sigma_\otimes \colon A \otimes B \to B \otimes A$, it follows that the following diagram also commutes:
  \setcounter{equation}{0}
  \begin{equation}\label{eq:decomp}
  \adjustbox{scale=0.62}{
  \begin{tikzcd}
    {(A \otimes (C \oplus D)) \oplus (B \otimes (C \oplus D)) } & {((C \oplus D) \otimes A) \oplus ((C \oplus D) \otimes B) } & {(C  \otimes A) \oplus (D \otimes A) \oplus (C \otimes B) \oplus (D  \otimes B) } \\
    {(A \oplus B) \otimes (C \oplus D)} \\
    & {} & {(A \otimes C) \oplus (A \otimes D) \oplus (B \otimes C) \oplus (B \otimes D)} \\
    {(C \oplus D) \otimes (A \oplus B)} \\
    {(C \otimes (A \oplus B)) \oplus (D  \otimes (A \oplus B))} & {((A \oplus B) \otimes C) \oplus ((A \oplus B) \otimes D)} & {(A \otimes C) \oplus (B \otimes C) \oplus (A \otimes D) \oplus (B \otimes D)}
    \arrow["{\sigma_\otimes \oplus \sigma_\otimes}", from=1-1, to=1-2]
    \arrow["{\delta_L \oplus \delta_L}", from=1-1, to=3-3]
    \arrow["{=}", from=1-2, to=1-3]
    \arrow["{\sigma_\otimes \oplus \sigma_\otimes \oplus \sigma_\otimes \oplus \sigma_\otimes}", from=1-3, to=3-3]
    \arrow["{=}", from=2-1, to=1-1]
    \arrow["{\sigma_\otimes}"', from=2-1, to=4-1]
    \arrow["{\delta_L}"', from=2-1, to=5-2]
    \arrow["{\id \oplus \sigma_\oplus \oplus \id}", from=3-3, to=5-3]
    \arrow["{=}", from=5-1, to=4-1]
    \arrow["{\sigma_\otimes \oplus \sigma_\otimes}", from=5-2, to=5-1]
    \arrow["{=}", from=5-3, to=5-2]
  \end{tikzcd}}
  \end{equation}
\end{proof}

\begin{lemma}\label{reduction}
  In a semisimple bipermutative category, every multiplicative symmetry
  $\sigma_\otimes \colon m \otimes n \to n \otimes m$ can be constructed as a composition of sums of identities and additive symmetries $\sigma_\oplus$.
\end{lemma}
\begin{proof}
  We proceed by well-founded induction on the pair $[m,n]$ under lexicographic ordering. Write $\sigma_{\otimes[m,n]}$ for the multiplicative symmetry $m \otimes n \to n \otimes m$.

  The base case $m=n=0$ is vacuous because then $\sigma_\otimes = \id[0]$.  
  For the induction step, assume that $\sigma_{\otimes[m',n']}$ is a composition of sums of identities and additive symmetries $\sigma_\oplus$ for all $(m',n')$ lexicographically before $(m,n)$. 
  By diagram~\eqref{eq:decomp} in the proof above, $\sigma_{\otimes[m,n]}$ is a composition of sums of identities and
  $\sigma_{\otimes[n-1,m]}$, $\sigma_{\otimes[1,m]}$,
  $\sigma_{\otimes[n-1,m-1]}$, $\sigma_{\otimes[n-1,1]}$,
  $\sigma_{\otimes[1,m-1]}$, $\sigma_{\otimes[1,1]}$,
  $\sigma_{\otimes[m-1,n]}$, and $\sigma_{\otimes[n,1]}$, as well as $\sigma_\oplus$.
  All these indices are lexicographically before $[n,m]$, except $[m-1,n]$.
  So for all multiplicative symmetries except $\sigma_{\otimes[m-1,n]}$, we can apply the induction hypothesis to write them as compositions of sums of identities and additive symmetries.

  To construct $\sigma_{\otimes[m-1,n]}$, notice that $\sigma_{\otimes[m-1,n]}^{-1} = \sigma_{\otimes[n,m-1]}$. Apply induction to construct $\sigma_{\otimes[n,m-1]}$ as a composition of sums of identities and additive symmetries. 
  As the inverse of identities and additive symmetries are, again, identities and additive summetries (respectively), this immediately lets us construct $\sigma_{\otimes[m-1,n]}$ of the right form.
\end{proof}

\begin{theorem}
  Semisimple bipermutative categories are PROPs, and a functor between semisimple bipermutative categories is a strict bipermutative functor if and only if it is a PROP morphism.
\end{theorem}
\begin{proof}
  By definition, the objects of semisimple bipermutative categories are
  precisely natural numbers. As regards morphisms, by distributivity
  every monoidal product of morphisms reduces to monoidal sums and
  multiplicative symmetries, and by \Cref{reduction}
  every coherence isomorphism of semisimple bipermutative categories
  -- and, consequently, every coherence condition of bipermutative
  categories -- can be expressed in terms of the strict symmetric
  monoidal structure $(\oplus,0)$.

  If $F \colon \cat{C} \to \cat{D}$ is a functor between semisimple bipermutative categories, it follows immediately that if $F$ is a strict bipermutative functor it also a PROP morphism. 
  Conversely, suppose $F$ is a PROP morphism. Since $F(n) = n$ on morphisms in particular $F(1) = 1$, so on objects it preserves the multiplicative unit. On morphisms it follows from \Cref{reduction} that $F(\sigma_\otimes) = \sigma_\otimes$ since $\sigma_\otimes$ is constructed (in both $\cat{C}$ and $\cat{D}$) purely in terms of the strict symmetric monoidal structure $(\oplus,0)$, which is preserved exactly by $F$. Moreover, $F$ preserves monoidal products strictly:
  \begin{align*}
    F(f \otimes g)
    & = F((f \otimes \id[n]) \circ (\id[m] \otimes g)) \\
    & = F(\sigma_\otimes \circ (\id[n] \otimes f) \circ \sigma_\otimes
    \circ (\id[m] \otimes g)) \\
    & = F(\sigma_\otimes \circ (\underbrace{f \oplus f \oplus \cdots
      \oplus f}_{n~\text{times}}) \circ \sigma_\otimes \circ
    (\underbrace{g \oplus g \oplus \cdots \oplus g}_{m~\text{times}})\\
    & = F(\sigma_\otimes) \circ F(\underbrace{f \oplus f \oplus \cdots
      \oplus f}_{n~\text{times}}) \circ F(\sigma_\otimes) \circ
    F(\underbrace{g \oplus g \oplus \cdots \oplus g}_{m~\text{times}})\\
    & = F(\sigma_\otimes) \circ \underbrace{F(f) \oplus F(f) \oplus \cdots
      \oplus F(f)}_{n~\text{times}} \circ \,F(\sigma_\otimes) \circ
    \underbrace{F(g) \oplus F(g) \oplus \cdots \oplus F(g)}_{m~\text{times}}\\
    & = \sigma_\otimes \circ \underbrace{F(f) \oplus F(f) \oplus \cdots
      \oplus F(f)}_{n~\text{times}} \circ \,{\sigma_\otimes} \circ
    \underbrace{F(g) \oplus F(g) \oplus \cdots \oplus F(g)}_{m~\text{times}}\\
    & = \sigma_\otimes \circ (\id[n] \otimes F(f)) \circ \,{\sigma_\otimes} \circ
    (\id[m] \otimes F(g)) \\
    & = (F(f) \otimes \id[n]) \circ (\id[m] \otimes F(g)) \\
    & = F(f) \otimes F(g)\text.
  \end{align*}
  Thus $F$ is a strict symmetric monoidal functor for $(\otimes,1)$ as
  well, in turn making it a strict bipermutative functor.
\end{proof}

Because free PROPs are known to exist~\cite{maclane:prop}, this proves \Cref{thm:freemodel} of the main article. 

\begin{proposition}\label{piprop}
  There exists a free bipermutative category $\prop{\zetak,S,V}$ on generators $\zetak \colon 1 \to 1$ and $V \colon 1 \oplus 1 \to 1 \oplus 1$ for $k \geq 2$. 
  Write $X$ for the symmetry morphism $\sigma \colon 1 \oplus 1 \to 1 \oplus 1$, and $S$ for $1 \oplus \zetak^{2^{k-2}} \colon 1 \oplus 1 \to 1 \oplus 1$.
  There exists a quotient $\Pi_k$ of $\prop{\zetak,S,V}$ by the smallest  congruence $\sim_k$ of bipermutative categories containing equations~\eqref{eq:axiom1}--\eqref{eq:axiom3}, and it is a bipermutative category.
\end{proposition}

\section*{Cancellativity}\label{cancellativity}

Note that throughout this section, we use $a,a',b,b',c$ to denote morphisms of
a category $\cat{C}$.

Next, we build towards soundness and completeness of $\Pi_k$. Our proof will use auxiliary objects by relying on the following property.

\begin{definition}
  A monoidal category $(\cat{C},\oplus,O)$ is \emph{strongly cancellative} when 
  \[
    a \oplus b=a' \oplus b'
    \implies
    a=a'\ \text{and}\ b=b'
  \]
  for all morphisms $a,a' \colon A \to A'$ and $b,b' \colon B \to B'$.
\end{definition}

Any monoidal category that embeds into one where $\oplus$ is a biproduct is strongly cancellative. This includes the PROP $\cat{Unitary}(R)$ of unitary matrices over an involutive ring $R$, which embeds into the rig category of $R$-modules, where $a \oplus b$ is a block-diagonal unitary matrix. 
We believe that $\Pi_k$ itself is strongly cancellative for all $k$ (it is known to be the case for $k=2$), but leave a proof of this as an open question.
All we need here is to force a PROP to become strongly cancellative.

\begin{proposition}\label{congruence}
  Any symmetric monoidal category $(\cat{C},\oplus,O)$ has a monoidal congruence defined for $a,a' \colon A \to A'$ by:
  \[
    a \approx a'
    \qquad \iff \qquad
    \exists b,b' \colon B \to B' \colon a \oplus b = a' \oplus b'
  \]
\end{proposition}
\begin{proof}
  We verify that the relation $\approx$ is a monoidal congruence:
  \begin{itemize}
    \item Reflexivity.
    Taking $b=b'=\id[I]$ shows that $a \approx a$.

    \item Symmetry.
    Suppose that $a \approx a'$, say because $a \oplus b = a' \oplus b'$. 
    Then also $a' \oplus b' = a \oplus b$, so $a' \approx a$.

    \item Transivity. 
    Suppose that $a \approx a'$ because $a \oplus b = a' \oplus b'$, and that $a' \approx a''$ because $a' \oplus c' = a'' \oplus c''$. Then $a \approx a''$ because
    \begin{align*}
      a \oplus (c' \circ b)
      &= (1 \oplus c') \circ (a \oplus b) \\ 
      &= (1 \oplus c') \circ (a' \oplus b')\\ 
      &= (a' \oplus c') \circ (1 \oplus b')\\ 
      &= (a'' \oplus c'') \circ (1 \oplus b') \\
      &= a'' \oplus (c'' \circ b')\text.
    \end{align*}

    \item Compositionality.
    Suppose that $a \approx a'$ because $a \oplus b = a' \oplus b'$ for $a,a' \colon A \to A'$, and that $c \approx c'$ because $c \oplus d = c' \oplus d'$ for $c,c' \colon A' \to A''$. Then $ca \approx c'a'$ because 
    \[
      ca \oplus db = (c \oplus d)(a \oplus b) = (c' \oplus d')(a' \oplus b') = c'a' \oplus d'b'\text.
    \]

    \item Monoidality.
    Suppose that $a \approx a'$ because $a \oplus b = a' \oplus b'$, and that $c \approx c'$ because $c \oplus d = c' \oplus d'$.
    Then
    \[
      a \oplus b \oplus c \oplus d = a' \oplus b' \oplus c' \oplus d'
    \]
    and by postcomposing with the symmetry hence also
    \[
      a \oplus c \oplus b \oplus d = a' \oplus c' \oplus b' \oplus d'\text,
    \]
    that is, $a \oplus c \approx a' \oplus c'$.
  \end{itemize}    
\end{proof}

\begin{corollary}
  If $(\cat{C},\oplus,O)$ is a symmetric monoidal category, then $\cat{C} \slash \mathop{\approx}$ is a well-defined strongly cancellative symmetric monoidal category. If $\cat{C}$ satisfies some equation, then so does $\cat{C}\slash\mathop{\approx}$.
\end{corollary}
\begin{proof}
  The first statement follows directly from \Cref{congruence}.
  The second statement follows from the fact that the (symmetric) monoidal functor $\cat{C} \to \cat{C}\slash\mathop{\approx}$ preserves equations.
\end{proof}

We can now define the functor of \Cref{thm:standardmodel} and \Cref{thm:density} of the main article.

\begin{proposition}\label{freeinterpretation}
  The following defines a symmetric monoidal functor $\sem{-} \colon \Pi_k \to \cat{Unitary}(\DD[\zetak])$:
  \[
    \sem{\zetak} = e^{2\pi i / 2^k} 
    \qquad\qquad
    \sem{V} = \frac{1}{2} \begin{pmatrix} 1+i & 1-i \\ 1-i & 1+i \end{pmatrix} 
  \]
\end{proposition}
\begin{proof}
  Observe that these matrices satisfy~\eqref{eq:axiom1}--\eqref{eq:axiom3}.
\end{proof}

Because $\cat{Unitary}(\DD[\zetak])$ is strongly cancellative, there is also a symmetric monoidal functor $\prop{\zetak,S,V} \to \cat{Unitary}(\DD[\zetak])$.

\begin{definition}
  \label{def:gates}
  For $k\geq 2$, define the following standard morphisms in $\Pi_k$:
  \begin{align*} 
    -1 &= \zetak^{2^{k - 1}} \colon 1 \to 1 &
    X &= \sigma_\oplus \colon 2 \to 2 \\
    H &= (\zetak^{2^{k - 3}})^7 \bullet T_k^{2^{k-2}}VT_k^{2^{k-2}} \colon 2 \to 2 \enspace (k \ge 3)&
    T_k &= \myid \oplus \zeta_k \colon 2 \to 2 
  \end{align*}
\end{definition}

The ring $\DD[\zetak]$ has a conjugation $x+\zetak y \mapsto (x+\zetak y)^* = x-\zetak y$, and this ring automorphism extends componentwise to an automorphism on the ring of matrices with entries in $\DD[\zetak]$. We now define a syntactic analogue.

\begin{definition}
  Let $a\in\Pi_k$ for $k \geq 2$. The \emph{conjugate} $a^*$ of $a$ is defined inductively by:
  \begin{itemize}
    \item $(\myid_n)^* = \myid_n$;
    \item $(\sigma)^* = \sigma$;
    \item $(\zeta_{k})^* = - \zeta_{k}$;
    \item $(V)^* = V$;
    \item $(ab)^* = (a)^*(b)^*$, for any $a,b \colon m \to m$;
    \item $(a\oplus b)^* = (a)^* \oplus (b)^*$, for any $a \colon m \to m$ and $b \colon n \to n$.
  \end{itemize}
\end{definition}

A straightforward induction shows that $\sem{a^*} = \sem{a}^*$.

\begin{definition}\label{catalyticembedding}
  Let $k\geq 3$. The \emph{catalytic embedding} $\Phi_k \colon
  \Pi_{k} \to \Pi_{k-1}$ is the symmetric monoidal functor defined on
  objects by $\Phi_k(n) = 2n$ and on morphisms by
  \begin{itemize}
    \item $\Phi_k(\myid_n) = \myid_n \oplus \myid_n$;
    \item $\Phi_k(\sigma) = \sigma\oplus \sigma$;
    \item $\Phi_k(\zeta_{k}) = X \circ (\myid \oplus \zeta_{k-1})$;
    \item $\Phi_k(V) = V\oplus V$;
    \item $\Phi_k(ab) = \Phi_k(a)\Phi_k(b)$ for $a,b\colon n \to n$;
    \item $\Phi_k(a\oplus b) =\sigma^{\otimes}_{2,m+n} \Phi_k(a) \oplus \Phi_k(b) \sigma^{\otimes}_{2,m+n}$, for $a \colon m \to m$ and $b \colon n \to n$.
  \end{itemize}
  For $n \in \N$, the \emph{catalyst} $c_{k,n} \colon 2n \to 2n$ is
  $
    c_{k,n} = (H \circ T_{k}) \otimes \myid_n
  $.
\end{definition}

\begin{lemma}\label{catalysissum}
  If $a \colon n \to n$ in $\Pi_k$ for $k \geq 3$, then
  $
    c_{k,n} \Phi_k(a) c_{k,n}^\dag \sim_{k} a \oplus a^*
  $.
\end{lemma}
\begin{proof}
  Proceed by structural induction on $a$. If $a=\myid_n$, then
  \begin{align*}
    c_{k,n} \circ \Phi_k(a) \circ c_{k,n}^\dagger
    &\sim_k c_{k,n} \circ \Phi_k(\myid_n) \circ c_{k,n}^\dagger \\
    &\sim_k (HT \otimes \myid_n) \circ \myid_n\oplus \myid_n \circ (T^\dagger H^\dagger \otimes \myid_n) \\
    &\sim_k (HT \otimes \myid_n) \circ \myid_2\otimes \myid_n \circ (T^\dagger H^\dagger \otimes \myid_n) \\    
    &\sim_k (HT \circ \myid_2 \circ T^\dagger H^\dagger) \otimes \myid_n\\
    &\sim_k \myid_2 \otimes \myid_n \\
    &\sim_k \myid_n \oplus \myid_n \\
    &\sim_k a \oplus a^*.
  \end{align*}
  If $a=\sigma_{m,n}$, then
  \begin{align*}
    c_{k,m+n} \circ \Phi_k(a) \circ c_{k,m+n}^\dagger
    &\sim_k c_{k,m+n} \circ \Phi_k(\sigma_{m,n}) \circ c_{k,m+n}^\dagger \\
    &\sim_k (HT \otimes \myid_{m+n}) \circ \sigma_{m,n}\oplus \sigma_{m,n} \circ (T^\dagger H^\dagger \otimes \myid_{m+n}) \\
    &\sim_k (HT \otimes \myid_{m+n}) \circ \myid_2 \otimes \sigma_{m,n} \circ (T^\dagger H^\dagger \otimes \myid_{m+n}) \\
    &\sim_k (HT \circ \myid_2 \circ T^\dagger H^\dagger) \otimes \sigma_{m,n}  \\
    &\sim_k \myid_2 \otimes \sigma_{m,n} \\
    &\sim_k \sigma_{m,n} \oplus \sigma_{m,n} \\
    &\sim_k a\oplus a^*
  \end{align*}
  If $a=\zeta_{k}$, then
  \begin{align*}
    c_{k,1} \circ \Phi_k(a) \circ c_{k,1}^\dagger
    &\sim_k c_{k,1} \circ \Phi_k(\zeta_{k}) \circ c_{k,1}^\dagger \\
    &\sim_k (HT \otimes \myid_1) \circ (X\circ (\myid \oplus \zeta_{k-1})) \circ (T^\dagger H^\dagger \otimes \myid_1) \\
    &\sim_k HT \circ (X\circ (\myid \oplus \zeta_{k-1})) \circ T^\dagger H^\dagger\\ 
    &\sim_k H\circ T \circ X \circ T^2 \circ T^\dagger \circ H^\dagger\\
    &\sim_k H\circ T \circ X \circ T \circ H^\dagger\\
    &\sim_k H\circ (\zeta_{k}\oplus \zeta_{k}) \circ X \circ H^\dagger\\
    &\sim_k (\zeta_{k}\oplus \zeta_{k}) \circ Z\\
    &\sim_k \zeta_{k}\oplus -\zeta_{k}    \\
    &\sim_k a\oplus a^*
  \end{align*}
  If $a=V$, then
  \begin{align*}
    c_{k,2} \circ \Phi_k(a) \circ c_{k,2}^\dagger
    &\sim_k c_{k,2} \circ \Phi_k(V) \circ c_{k,2}^\dagger \\
    &\sim_k (HT \otimes \myid_2) \circ V\oplus V \circ (T^\dagger H^\dagger \otimes \myid_2) \\
    &\sim_k (HT \otimes \myid_2) \circ \myid_2 \otimes V \circ (T^\dagger H^\dagger \otimes \myid_2) \\
    &\sim_k (HT \circ \myid_2 \circ T^\dagger H^\dagger) \otimes V  \\
    &\sim_k \myid_2 \otimes V \\
    &\sim_k V \oplus V\\
    &\sim_k a\oplus a^*
  \end{align*}
  If $a=a_1a_2$, for $a_1,a_2 \colon m \to m$ in $\Pi_k$, then by induction,
  \begin{align*}
    c_{k,m} \circ \Phi_k(a) \circ c_{k,m}^\dagger
    &\sim_k c_{k,m} \circ \Phi_k(a_1a_2) \circ c_{k,m}^\dagger \\
    &\sim_k c_{k,m} \circ \Phi_k(a_1)\Phi_k(a_2) \circ c_k^\dagger \\    
    &\sim_k c_{k,m} \circ \Phi_k(a_1)\circ c_{k,m} \circ c_{k,m}^\dagger \Phi_k(a_2) \circ c_k^\dagger \\        
    &\sim_k a_1 \oplus a_1^* \circ a_2 \oplus a_2^* \\
    &\sim_k (a_1a_2)\oplus (a_1^*a_2^*) \\
    &\sim_k a\oplus a^*
  \end{align*}
  Finally, consider the case $a=a_1\oplus a_2$, for $a_1 \colon m \to m$ and $a_2 \colon n \to n$ in $\Pi_k$. Observe
  \begin{align*}
    c_{k,m+n} \circ \sigma^{\otimes}_{2,m+n} 
    &\sim_k (HT \otimes \myid_{m+n})\circ \sigma^{\otimes}_{2,m+n} \\ 
    &\sim_k \sigma^{\otimes}_{2,m+n}\circ (\myid_{m+n} \otimes HT) \\ 
    &\sim_k \sigma^{\otimes}_{2,m+n}\circ ((\myid_{m}\oplus \myid_{n}) \otimes HT) \\     
    &\sim_k \sigma^{\otimes}_{2,m+n}\circ ((\myid_{m} \otimes HT) \oplus (\myid_{n} \otimes HT)) \\ 
    &\sim_k \sigma^{\otimes}_{2,m+n}\circ (c_{k,m} \oplus c_{k,n}).
  \end{align*}
  Hence, by induction,
  \begin{align*}
    c_{k,m+n} \circ \Phi_k(a) \circ c_{k,m+n}^\dagger
    &\sim_k c_{k,m+n} \circ \Phi_k(a_1\oplus a_2) \circ c_{k,m+n}^\dagger \\
    &\sim_k c_{k,m+n}\circ (\sigma^{\otimes}_{2,m+n} \Phi_k(a) \oplus \Phi_k(b) \sigma^{\otimes}_{2,m+n}) \circ c_{k,m+n}^\dagger\\ 
    &\sim_k \sigma^{\otimes}_{2,m+n} (c_{k,m} \oplus c_{k,n})  \Phi_k(a) \oplus \Phi_k(b) (c_{k,m}^\dagger \oplus c_{k,n}^\dagger) \sigma^{\otimes}_{2,m+n} \\
    &\sim_k \sigma^{\otimes}_{2,m+n} (a_1\oplus a_1^*) \oplus (a_2\oplus a_2^*) \sigma^{\otimes}_{2,m+n}    \\
    &\sim_k (a_1\oplus a_2)\oplus (a_1^*\oplus a_2^*) \\
    &\sim_k a \oplus a^*. 
  \end{align*}
\end{proof}

We can now prove \Cref{thm:precision} from the main article.

\begin{lemma}\label{lower}
  If $k\geq 3$, and $a,b \in \Pi_{k}$ satisfy $\sem{a}=\sem{b}$,
  then $\sem{\Phi_k(a)}=\sem{\Phi_k(b)}$.
\end{lemma}
\begin{proof}
  Let $a,b \in \Pi_{k}$ and suppose that $\sem{a}=\sem{b}$, By \Cref{catalysissum}, 
  \[
  c_k \Phi_k(a) c_k^\dagger \sim_k a\oplus a^*
  \quad
  \mbox{and}
  \quad
  c_k \Phi_k(b) c_k^\dagger \sim_k b\oplus b^*. 
  \]
  Because the axioms are sound with respect to the interpretation, then
 \begin{align*}
 \sem{c_k}\sem{ \Phi_k(a)}\sem{c_k^\dagger}
  &= \sem{c_k \Phi_k(a) c_k^\dagger}  \\
  &= \sem{ a\oplus a^*} \\
  &= \sem{a}\oplus\sem{a}^* \\
  &= \sem{b}\oplus\sem{b}^*\\
  &= \sem{b\oplus b^*}  \\
  &= \sem{c_k \Phi_k(b) c_k^\dagger}\text.
  \end{align*}
  We conclude that
  \[
  \sem{\Phi_k(a)} 
  = \sem{c_k^\dagger} \sem{c_k}\sem{ \Phi_k(a)}\sem{c_k^\dagger} \sem{c_k}
  = \sem{c_k^\dagger} \sem{c_k}\sem{ \Phi_k(b)}\sem{c_k^\dagger} \sem{c_k}
  = \sem{\Phi_k(b)} . \qedhere
  \]  
\end{proof}

Write $\approx_k$ for the smallest congruence of bipermutative categories
containing $\sim_k$ and $\approx$.  It is clear that if $a \approx_k
b$ for $a,b \in \Pi_k$, then $\sem{a}=\sem{b}$, that is, the equations
of $\Pi_k$ are \emph{sound}. The following theorem shows
\emph{completeness}, which is the converse.

\begin{theorem}
  \label{thm:completeness}
  Let $k\geq 2$. If $a,b \in \Pi_k$ satisfy $\sem{a}=\sem{b}$, then $a
  \approx_k b$.
\end{theorem}

\begin{proof}
  The proof proceeds by induction on $k$. The base case $k=2$ is
  known~\cite{caretteetal:sqrtpi}. Now assume that the property holds
  for $k$, and suppose that $a,b \in \Pi_{k+1}$ satisfy
  $\sem{a}=\sem{b}$.  Then $\Phi_k(a),\Phi_k(b)\in \Pi_k$ and, by
  \Cref{lower}, $\sem{\Phi_{k}(a)}=\sem{\Phi_{k}(b)}$.  Therefore, by
  the induction hypothesis, $\Phi_{k}(a) \approx_k \Phi_{k}(b)$.
  Hence, $\Phi_{k}(a) \approx_{k+1} \Phi_{k+1}(b)$, and thus $C
  \Phi_{k}(a) C^\dag \approx_{k+1} C \Phi_{k+1}(b) C^\dag$.
  \Cref{catalysissum} now provides $a',b' \in \Pi_{k+1}$ satisfying
  \[
  a \oplus a' \approx_{k+1} C \Phi_{k+1}(a) C^\dag \approx_{k+1} C
  \Phi_{k+1}(b) C^\dag \approx_{k+1} b \oplus b'\text.
  \]
  Thus $a \approx_{k+1} b$ by strong cancellation.
\end{proof}

The goal of the rest of this Supporting Information is to prove completeness, in \Cref{completeness} below, of the free model including measurement. The main idea is to reduce to an existing algebraic axiomatisation~\cite{staton:quantum}. This is achieved in four stages. After defining the free model of measurement, we make precise how a model can validate equations. For this to make sense, we then show that the free model including measurement supports classical control, by first reducing to manipulations of injections between finite sets. Finally, we can then verify the required axioms.

Measurement can be freely added to any model in three steps, as follows~\cite{heunenkaarsgaard:informationeffects,andresmartinezheunenkaarsgaard:universal}. Let $\cat{C}$ be a rig category. 
\begin{enumerate}
  \item 
  Freely add the ability to initialise states. This results in a category $R(\cat{C})$ that has the same objects as $\cat{C}$, but morphisms $A \to B$ in $R(\cat{C})$ are morphisms $A \oplus H \to B$ in $\cat{C}$ for some \emph{heap} object $H$, where morphisms are identified when they are equal up to preprocessing the heap. There is a unique morphism $O \to A$ in $R(\cat{C})$ for each object $A$, that initialises states. If $\cat{C}$ is the PROP of unitary matrices, then $R(\cat{C})$ is the category of isometries.

  \item 
  Freely add decoherence. This results in a category $\LR(\cat{C})$ that has the same objects as $\cat{C}$, but morphisms $A \to B$ in $\LR(\cat{C})$ are morphisms $A \to B \otimes G$ in $R(\cat{C})$ for some \emph{garbage} object $G$, where morphisms are identified when they are equal up to postprocessing the garbage. There is a unique morphism $A \to I$ in $\LR(\cat{C})$ for each object $A$, that discards the system $A$ by removing it from the control of the experimenter to the environment, according to Zurek's interpretation of decoherence~\cite{zurek:decoherence}.
  If $\cat{C}$ is the PROP of unitary matrices, then $\LR(\cat{C})$ is the category of finite-dimensional Hilbert spaces and completely positive linear maps, also known as quantum channels.

  \item 
  Freely add classical control, \textit{i.e.}\ the ability to branch on measurement outcomes. This results in a category $\Split(\LR(\cat{C}))$, where objects are projections $p^\dag = p^2 = p \colon A \to A$ in $\LR(\cat{C})$, and where morphisms $p \to p'$ are $f \colon A \to A'$ in $\LR(\cat{C})$ satisfying $p'fp=f$. 
  This is also known as the \emph{Cauchy completion} or \emph{Karoubi envelope}~\cite{selinger:idempotent}, and has the universal property that projections on $B$ in $\LR(\cat{C})$ correspond to embeddings $A \to B$ in $\Split(\LR(\cat{C}))$.
  If $\cat{C}$ is the PROP of unitary matrices, then $\Split(\LR(\cat{C}))$ is the category of finite-dimensional C*-algebras and quantum channels.
\end{enumerate}
The $\oplus$ of the original rig category $\cat{C}$ becomes a coproduct in $\Split(\LR(\cat{C}))$, and for this to happen there is no choice in the order in which these steps are taken~\cite{heunenkaarsgaard:informationeffects,andresmartinezheunenkaarsgaard:universal}.

\section*{Algebraic theories}\label{sec:structures}
An important step in the axiomatisation of quantum computation is Staton's algebraic theory of quantum computation~\cite{staton:quantum}. This shows that the theory of finite-dimensional C*-algebras and quantum channels can be reduced to a small number of generating spaces, channels, and equations between them. In brief, an algebraic theory consists of two things:
\begin{itemize}
  \item Some \emph{operations}, each carrying an arity $p \to (n_1,
    \dots, n_m)$, meaning that the operation accepts $p$ parameters and produces a number of results in the list $n_1, \dots, n_m$. Operations can be combined
    both in sequence (when the number of parameters match) and in parallel
    to form \emph{terms} $t$, which describe computations.
  \item Some \emph{equations} describing which terms are equal.
\end{itemize}
As the name suggests, this setup is entirely analogous to that found
in algebra -- for example in combinatorial group theory, where groups
are described in terms of a number of elementary operations
(\emph{generators}) and equations between elements of the group
(\emph{relations}).

The algebraic theory of quantum computation consists of three
operations (and twelve equations, which we disregard for now). The
operations are:
\begin{itemize}
  \item $\mathsf{new} \colon 0 \to 1$ which allocates a new qubit in the state $\ket{0}$.
  \item $\mathsf{measure} \colon 1 \to (0,0)$ which measures a qubit in
    the computational basis and records the outcome in the choice of output.
  \item $\mathsf{apply}_U \colon n \to n$ which applies an $n$-qubit unitary $U$.
\end{itemize}
An algebraic theory can be interpreted in a category that has enough structure to support the operations and equations:
\begin{enumerate}
\item It must have coproducts, it must have a distinguished object $X$, and for 
  each operation $O \colon n \to n_1, \dots, n_m$ a morphism
  $\sem{O} \colon n \bullet X \to (n_1 \bullet X + \cdots + n_m \bullet X)$ for an
 \emph{action} $n \bullet -$ of
  the PROP $\cat{FinPerm}$ whose morphisms $n \to n$ are the permutations of $\{1,\ldots,n\}$ (and there are no morphisms $n \to n'$ if $n\neq n'$)~\cite{janelidzekelly:action}.
\item For each equation $s=t$, it must be the case that the morphisms
  $\sem{s}$ and $\sem{t}$ coincide: $\sem{s} = \sem{t}$.
\end{enumerate}
Such a category is called a \emph{model} of the algebraic theory.

The algebraic theory of quantum computation has a natural model given
by the category of finite-dimensional C*-algebras and quantum
channels.  Here, a term $t$ of arity $(p \mid n_1, \dots, n_m)$
describes a quantum channel $\sem{t}$ from the algebra
$M_p(\mathbb{C})$ of $2^p$-by-$2^p$ matrices to the direct sum of matrix algebras
$M_{n_1}(\mathbb{C}) \oplus \cdots \oplus
M_{n_m}(\mathbb{C})$. Equations $t=u$ are assertations of the fact
that the quantum channels $\sem{t}$ and $\sem{u}$ do the same thing,
i.e., $\sem{t} = \sem{u}$. What is surprising is that this model is
\emph{complete}: if $\sem{t} = \sem{u}$ happens to be the case (for
any choice of quantum channel $t$ and $u$), then this can be
established purely by applying the equations of the algebraic theory.

Our goal is to prove that $\Split(\LR(\Pi_k))$ is also a complete model of Staton's algebraic theory of quantum computation. For that to make sense, we first explain in full detail how categories can interpret algebraic theories.

For a category $\cat{C}$ to host a structure for the signature of an
algebraic theory, it must have an \emph{action} for $\cat{FinPerm}$:
there must be a functor $\bullet \colon \cat{FinPerm} \times \cat{C}
\to \cat{C}$ and natural isomorphisms $0 \bullet X \cong X$ and $(n+m)
\bullet X \cong n \bullet (m \bullet X)$ subject to certain coherence
conditions. We start by showing that any symmetric monoidal category
$\cat{C}$ gives rise to a family of canonical actions for
$\cat{FinPerm}$.

Let $A^{\otimes n}$ denote the $n$-fold monoidal product of an object
$A$ with itself (choosing a canonical bracketing, e.g., leaning all
parentheses to the left), with $A^{\otimes 0} = I$.
Define the \emph{exponential action} of an object $S$ in $\cat{C}$ to
be the functor $\bullet \colon \cat{FinPerm} \times \cat{C} \to \cat{C}$ given
by $n \bullet X = S^{\otimes n} \otimes X$ on objects, and on
morphisms $(\pi : n \to n) \bullet (f \colon X \to Y) = \lceil\pi\rceil
\otimes f$, where $\lceil\pi\rceil \colon S^{\otimes n} \to S^{\otimes n}$
is the unique permutation that uses the natural symmetry $\sigma_{X,Y}
\colon X \otimes Y \to Y \otimes X$ to permute the $n$ different $S$-wires
according to the permutation $\pi \colon n \to n$.
\begin{lemma}\label{lem:expact}
  When $(\cat{C},\otimes,I)$ is a symmetric monoidal category and $S$
  is an object of $\cat{C}$, the \emph{exponential action of $S$ on
  $\otimes$} is an action of $\cat{FinPerm}$ on $\cat{C}$.
\end{lemma}
\begin{proof}
  To see that $\bullet : \cat{FinPerm} \times \cat{C} \to \cat{C}$ is
  bifunctorial, we notice that
  \begin{equation*}
    \mathrm{id}_n \bullet \mathrm{id}_X =
    \lceil\mathrm{id}_n\rceil \otimes \mathrm{id}_X =
    \mathrm{id}_{S^{\otimes n}} \otimes \mathrm{id}_X
    = \mathrm{id}_{S^{\otimes n} \otimes X} = \mathrm{id}_{n \bullet X}
  \end{equation*}
  and
  \begin{align*}
    (\pi_2 \bullet g) \circ (\pi_1 \bullet f) & =
    (\lceil\pi_2\rceil \otimes g) \circ (\lceil\pi_1\rceil \otimes f) \\
    & = (\lceil\pi_2\rceil \circ \lceil\pi_1\rceil) \otimes (g \circ f) \\
    & = \lceil\pi_2 \circ \pi_1\rceil \otimes (g \circ f) \\
    & = (\pi_2 \circ \pi_1) \bullet (g \circ f)
  \end{align*}
  using the fact that permuting $n$ wires according to first $\pi_1$ and then
  $\pi_2$ is exactly the same as permuting them according to their composition
  $\pi_2 \circ \pi_1$. To see that this is an action, we notice first that
  \begin{equation*}
    0 \bullet X = S^{\otimes 0} \otimes X = I \otimes X \cong X
  \end{equation*}
  and
  \begin{align*}
    (n+m) \bullet X & = S^{\otimes n+m} \otimes X \\
    & \cong S^{\otimes n} \otimes (S^{\otimes m} \otimes X) \\
    & = n \bullet (m \bullet X)
  \end{align*}
  where the isomorphism $S^{\otimes n+m} \otimes X \cong S^{\otimes n}
  \otimes (S^{\otimes m} \otimes X)$ uses associativity to move
  parentheses accordingly. As for the coherence conditions, we may
  freely assume that $\cat{C}$ is strict symmetric monoidal, and
  notice that $\cat{FinPerm}$ already is strict symmetric monoidal.
  In this strict case, we notice that the isomorphisms $\vartheta :
  (n+m) \bullet X \to n \bullet (m \bullet X)$ and $\nu : 0 \bullet X
  \to X$ given above, as well as all coherence isomorphisms of
  $\cat{C}$ and $\cat{FinPerm}$, are identities, so the three
  coherence conditions~\cite{janelidzekelly:action} trivially
  commute (by virtue of every morphism involved being the identity).
\end{proof}

We are now ready to define structures for a signature. Instead of the Heisenberg picture~\cite{staton:quantum}, we will use the Schrödinger picture.
For the following definition to make sense, there must be coproducts, which we will avail ourselves in \cref{prop:coproducts} below.

\begin{definition}
  A \emph{signature} is a set of operations with arity $O \colon p \to (m_1,\ldots,m_k)$, intuitively specifying that operation $O$ takes $p$ many inputs and can produce either $m_1$, $m_2$, \ldots, or $m_k$ many outputs.
  Let $\cat{C}$ have an action of $\cat{FinPerm}$. A \emph{structure}
  for a signature in $\cat{C}$ is an object $X$ together with, for
  each operation $O \colon p \to (m_1, \ldots, m_k)$ a morphism $p \bullet X
  \to (m_1 \bullet X) + \cdots + (m_k \bullet X)$.
\end{definition}

\section*{The free coaffine rig category}
\label{subsec:fininjinitial}

Many of the equations we will be interested in are easily seen to hold in the PROP $\cat{FinInj}$ whose morphisms $m \to n$ are injections $\{1,\ldots,m\} \rightarrowtail \{1,\ldots,n\}$. The rig structure is given on objects by
addition and multiplication of natural numbers. On morphisms $f \colon n
\to m$ and $g \colon p \to q$, their sum $f \oplus g \colon n+p \to m+q$ is
$$
(f \oplus g)(a) = \begin{cases}
  f(a) & \text{if } a \leq n \\
  m+g(a-n) & \text{otherwise}
\end{cases}
$$
while their product $f \otimes g \colon np \to mq$ is given by the
function 
\begin{align*}
p_{f,g} \colon n \times p & \to m \times q \\
(a,b) &\mapsto (f(a),g(b))
\end{align*}
conjugated by the Cantor permutation $\{1,\ldots,np\} \to \{1,\ldots,n\} \times \{1,\ldots,p\}$.

\begin{lemma}\label{lem:fininj}
  $R(\cat{FinPerm})$ is rig isomorphic to $\cat{FinInj}$.
\end{lemma}
\begin{proof}
  It follows by the universal property of the $R$-construction that
  the forgetful functor $F \colon \cat{FinPerm} \to \cat{FinInj}$ factors
  as
\[\begin{tikzcd}
  {\mathbf{FinPerm}} & {R(\mathbf{FinPerm})} \\
  & {\mathbf{FinInj}}
  \arrow["{\mathcal{E}}", from=1-1, to=1-2]
  \arrow["F"', from=1-1, to=2-2]
  \arrow["{\hat{F}}", dashed, from=1-2, to=2-2]
\end{tikzcd}\]
for a unique rig functor $\hat{F}$. Since $F$ and $\mathcal{E}$ are
both identity on objects, it follows that $\hat{F}$ is identity (so
specifically bijective) on objects as well: as such, $\hat{F}$ is an
isomorphism of rig categories iff it is full and faithful.

By the Expansion-Raw Morphism factorisation~\cite[Lemma 2.8]{hermidatennent:indeterminates}, every morphism in $R(\mathbf{FinPerm})$ factors uniquely
as a canonical injection $\amalg_1 \colon n \to n + k$ followed by a bijection
$\pi \colon n + k \to m$. It follows that $\hat{F}$ is full
iff every injection between finite sets admits such a factorisation,
and faithful iff this factorisation is unique up to a bijection
applied solely to $k$.

Let $f \colon n \to m$ be some injection, and consider the set $m \setminus
\mathrm{im}(f)$ with $k = |m \setminus \im(f)|$. Notice that $m
\setminus \im(f)$ is a finite subset of an ordered set $m$, so ordered
again. This means that we can define a bijection $\psi \colon k \to m
\setminus \im(f)$ by sending $\psi(1)$ to the least element of $m
\setminus \im(f)$, $\psi(1)$, to the second-to-least element, and so
on. Now define $\pi \colon n + k \to m$ by
$$
\pi(a) = \begin{cases}
  f(a) & \text{if }a \leq n \\
  \psi(a-n) & \text{otherwise}
\end{cases}
$$
This is injective since $f$ and $\psi$ are injective and their images
are disjoint by
$$
\im(f) \cap \im(\psi) = \im(f) \cap (m \setminus \im(f)) = \emptyset,
$$
and it is surjective since
$$
\im(\pi) = \im(f) \cup \im(\psi) = \im(f) \cup (m \setminus \im(f)) = m
$$
so bijective. But then since $\amalg_1 \colon n \to n+k$ is given by
$\amalg_1(n) = n$, it follows by definition of $\pi$ that $\pi \circ
\amalg_1 = f$. Thus $\hat{F}$ is full.

To see that it is also faithful, suppose there exists some other $\pi'
\colon n+k \to m$ such that $\pi' \circ \amalg_1 = f$ yet $\pi \neq
\pi'$. Since $\pi' \circ \amalg_1 = f = \pi \circ \amalg_1$, it
follows for all $a \leq n$ that
$$\pi(a) = \pi(\amalg_1(a)) = f(a) = \pi'(\amalg_1(a)) = \pi'(a)$$ so
$\pi$ and $\pi'$ can only disagree on values of $a$ where $n \le a \le
n+k$, which must either way fall in the subset of the image $m
\setminus \im(f)$. Consider now the bijection $\pi'^{-1} \circ \pi \colon
n+k \to n+k$. By the argument above $(\pi'^{-1} \circ \pi)(a) = a$
when $a \leq n$, and $(\pi'^{-1} \circ \pi)(a) \ge n$ when $a \ge n$.
Defining $\psi'(a) = (\pi'^{-1} \circ \pi)(a+n)-n$, we see that
when $a \leq n$, 
$$
(\myid_n \oplus \psi')(a) = \myid_n(a) = a = (\pi'^{-1} \circ \pi)(a)
$$
and when $a > n$ we have
$$
(\myid_n \oplus \psi')(a) = n+\psi'(a-n) = n+(\pi'^{-1} \circ \pi)((a-n)+n)-n
= (\pi'^{-1} \circ \pi)(a)
$$
so $\myid \oplus \psi' = \pi'^{-1} \circ \pi$. But then unicity up to
a bijection applied on $k$ follows by
$$\pi' \circ (\myid_n \oplus \psi') = \pi' \circ (\pi'^{-1} \circ \pi)
= (\pi' \circ \pi'^{-1}) \circ \pi = \pi$$
which, in turn, means that $\hat{F}$ is faithful as well.
\end{proof}

A rig category is \emph{coaffine} when any object $A$ has a unique morphism from the additive unit $O$.

\begin{theorem}
  If $\cat{C}$ is an coaffine rig category, there exists a unique
  rig functor $\cat{FinInj} \to \cat{C}$.
\end{theorem}
\begin{proof}
  Since $\cat{FinPerm}$ is the initial rig category there exists a
  unique rig functor $\cat{FinPerm} \to \cat{C}$. Since $\cat{C}$ is
  coaffine, it follows by the universal property of the
  $R$-construction, that this unique functor factors uniquely as
\[\begin{tikzcd}
  {\mathbf{FinPerm}} && {R(\mathbf{FinPerm})} \\
  && {\mathbf{C}}
  \arrow["{\mathcal{E}}", from=1-1, to=1-3]
  \arrow["{!}"', dashed, from=1-1, to=2-3]
  \arrow["{\hat{!}}", dashed, from=1-3, to=2-3]
\end{tikzcd}\]
This takes care of existence. For uniqueness, suppose there exists
another rig functor $F \colon R(\cat{FinPerm}) \to \cat{C}$. We note that,
by uniqueness of the rig functor ${!} \colon \cat{FinPerm} \to \cat{C}$, we
must also have ${!} = F \circ \mathcal{E}$. By the Expansion-Raw
Morphism factorisation, every morphism in $R(\cat{FinPerm})$ factors as
a canonical injection $\amalg_1 \colon n \to n + k$ followed by a bijection
$\pi \colon n+k \to m$. By ${!} = F \circ \mathcal{E}$, $F$ and $\hat{!}$
must agree on all bijections, so they can only every disagree on where
they send the canonical injections $\amalg_i$. Since $F$ and $\hat{!}$
are rig functors they must preserve the initial object, as it is the
additive unit in both $R(\cat{FinPerm})$ and $\cat{C}$: thus, they
must also preserve initial maps $0 \to A$, as these are unique in both
$R(\cat{FinPerm})$ and $\cat{C}$. But then they must preserve
$\amalg_1$ since $\amalg_1 = (\id \oplus {!}) \circ \rho^{-1}_\oplus$, where $\rho_\oplus \colon O \oplus A \to A$ is the coherence isomorphism, and
$F$ and $\hat{!}$ preserve all maps involved; so $F = \hat{!}$.

Since $R(\cat{FinPerm})$ and $\cat{FinInj}$ are rig isomorphic
by Lemma~\ref{lem:fininj}, it follows that $\cat{FinInj}$ satisfies
the same universal property.
\end{proof}

Another way of phrasing the theorem above is that $\cat{FinInj}$ is the free coaffine rig category.
A diagram is called coaffine when it is made up of rig
coherence isomorphisms and initial maps $0 \to A$.

\begin{corollary}\label{corr:coaffinediagrams}
  An coaffine diagram that commutes in $\cat{FinInj}$
  commutes in any other coaffine rig category as well.
\end{corollary}

\section*{Classical control}\label{sec:coproducts}

Our last preparation before verifying the required axioms is to show that $\Split(\LR(\cat{C}))$ has classical control; more precisely, that $\oplus$ is a coproduct. The main idea is that in an coaffine rig category like $\Split(\LR(\cat{C}))$, we can record the outcome of a measurement in the object $I \oplus I$ as follows.

\begin{definition}
In a coaffine rig category, define a natural family of morphisms $\mu_{A,B}$:

\[\begin{tikzcd}
  {A \oplus B} &&& {(A \oplus B) \otimes (I \oplus I)} \\
  {(A \otimes I) \oplus (B \otimes I)} &&& {(A \otimes (I \oplus I)) \oplus (B \otimes (I \oplus I))}
  \arrow["{\mu_{A,B}}", from=1-1, to=1-4]
  \arrow["{\rho_\otimes^{-1} \oplus \rho_\otimes^{-1} }"', from=1-1, to=2-1]
  \arrow["{(\mathrm{id} \otimes \amalg_1) \oplus (\mathrm{id} \otimes \amalg_2)}"', from=2-1, to=2-4]
  \arrow["{\delta_R^{-1}}"', from=2-4, to=1-4]
\end{tikzcd}\]
\end{definition}

\begin{lemma}
   The pentagon below commutes in any coaffine rig category.
 \[\begin{tikzcd}[column sep=tiny]
  & {(A \oplus B) \otimes ((I \oplus I) \otimes (I \oplus I))} \\
  {(A \oplus B) \otimes (I \oplus I)} && {((A \oplus B) \otimes (I \oplus I)) \otimes (I \oplus I)} \\
  {A \oplus B} && {(A \oplus B) \otimes (I \oplus I)}
  \arrow["{\mathrm{id} \otimes \mu_{I,I}}", from=2-1, to=1-2]
  \arrow["{\alpha_\otimes}"', from=2-3, to=1-2]
  \arrow["{\mu_{A,B}}", from=3-1, to=2-1]
  \arrow["{\mu_{A,B}}"', from=3-1, to=3-3]
  \arrow["{\mu_{A,B} \otimes \mathrm{id}}"', from=3-3, to=2-3]
 \end{tikzcd}\]
 \end{lemma}
 \begin{proof}
   By commutativity of the diagram below.
\[\adjustbox{scale=0.50}{\begin{tikzcd}[ampersand replacement=\&]
  {A \oplus B} \& {(A \otimes I) \oplus (B \otimes I)} \& {\big(A \otimes (I \oplus I)\big) \oplus \big(B \otimes (I \oplus I)\big)} \& {(A \oplus B) \otimes (I \oplus I)} \\
  \&\& {\big((A \otimes I) \otimes (I \oplus I)\big) \oplus \big((B \otimes I) \otimes (I \oplus I)\big)} \& {\big((A \otimes I) \oplus (B \otimes I)\big) \otimes (I \oplus I)} \\
  \&\& {\big((A \otimes (I \oplus I)) \otimes (I \oplus I)\big) \oplus \big((B \otimes (I \oplus I)) \otimes (I \oplus I)\big)} \& {\big((A \otimes (I \oplus I)) \oplus (B \otimes (I \oplus I))\big) \otimes (I \oplus I)} \\
  \& {\big(A \otimes (I \otimes I)\big) \oplus \big(B \otimes (I \otimes I)\big)} \& {\big(A \otimes ((I \oplus I) \otimes (I \oplus I))\big) \oplus \big(B \otimes ((I \oplus I) \otimes (I \oplus I))\big)} \& {\big((A \oplus B) \otimes (I \oplus I)\big) \otimes (I \oplus I)} \\
  \&\&\& {(A \oplus B) \otimes \big((I \oplus I) \otimes (I \oplus I)\big)} \\
  {(A \otimes I) \oplus (B \otimes I)} \&\& {\big(A \otimes (I \oplus I)\big) \oplus \big(B \otimes (I \oplus I)\big)} \& {(A \oplus B) \otimes (I \oplus I)}
  \arrow["{\rho_\otimes^{-1} \oplus \rho_\otimes^{-1}}", from=1-1, to=1-2]
  \arrow[""{name=0, anchor=center, inner sep=0}, "{\rho_\otimes^{-1} \oplus \rho_\otimes^{-1}}"', from=1-1, to=6-1]
  \arrow["{(\mathrm{id} \otimes \amalg_1) \oplus (\mathrm{id} \otimes \amalg_2)}", from=1-2, to=1-3]
  \arrow[""{name=1, anchor=center, inner sep=0}, "{\rho_\otimes^{-1} \oplus \rho_\otimes^{-1}}"', from=1-2, to=4-2]
  \arrow[""{name=2, anchor=center, inner sep=0}, "{\delta_R^{-1}}", from=1-3, to=1-4]
  \arrow["{(\rho_\otimes^{-1} \otimes \mathrm{id}) \oplus (\rho_\otimes^{-1} \otimes \mathrm{id})}"', from=1-3, to=2-3]
  \arrow["{(\rho_\otimes^{-1} \oplus \rho_\otimes^{-1}) \otimes \mathrm{id}}", from=1-4, to=2-4]
  \arrow[""{name=3, anchor=center, inner sep=0}, "{\delta_R^{-1}}"', from=2-3, to=2-4]
  \arrow["{\big((\mathrm{id} \otimes \amalg_1) \otimes \mathrm{id}\big) \oplus \big((\mathrm{id} \otimes \amalg_2) \otimes \mathrm{id}\big)}"', from=2-3, to=3-3]
  \arrow["{\big((\mathrm{id} \otimes \amalg_1) \oplus (\mathrm{id} \otimes \amalg_2)\big) \otimes \mathrm{id}}", from=2-4, to=3-4]
  \arrow[""{name=4, anchor=center, inner sep=0}, "{\delta_R^{-1}}"', from=3-3, to=3-4]
  \arrow["{\alpha_\otimes \oplus \alpha_\otimes}"', from=3-3, to=4-3]
  \arrow["{\delta_R^{-1} \otimes \mathrm{id}}", from=3-4, to=4-4]
  \arrow["{\big(\mathrm{id} \otimes (\amalg_1 \otimes \amalg_1)\big) \oplus \big(\mathrm{id} \otimes (\amalg_2 \otimes \amalg_2)\big)}"', shift right=2, draw=none, from=4-2, to=4-3]
  \arrow[from=4-2, to=4-3]
  \arrow[""{name=5, anchor=center, inner sep=0}, "{\delta_R^{-1}}"', from=4-3, to=5-4]
  \arrow["{\alpha_\otimes}", from=4-4, to=5-4]
  \arrow["{(\mathrm{id} \otimes \amalg_1) \oplus (\mathrm{id} \otimes \amalg_2)}"', from=6-1, to=6-3]
  \arrow["{(\mathrm{id} \otimes \mu_{I,I}) \oplus (\mathrm{id} \otimes \mu_{I,I})}"', from=6-3, to=4-3]
  \arrow[""{name=6, anchor=center, inner sep=0}, "{\delta_R^{-1}}"', from=6-3, to=6-4]
  \arrow["{\mathrm{id} \otimes \mu_{I,I}}"', from=6-4, to=5-4]
  \arrow["{(i)}"{description}, draw=none, from=0, to=1]
  \arrow["{(ii)}"{description}, draw=none, from=1, to=3-3]
  \arrow["{(iii)}"{description}, draw=none, from=2, to=3]
  \arrow["{(iii)}"{description}, draw=none, from=3, to=4]
  \arrow["{(iv)}"{description}, draw=none, from=5, to=4]
  \arrow["{(iii)}"{description}, draw=none, from=5, to=6]
\end{tikzcd}}\]
 Here, $(i)$ commutes by Corollary~\ref{corr:coaffinediagrams},
 $(ii)$ by bifunctoriality of $\otimes$ and coherence, $(iii)$ by
 naturality of $\delta_R^{-1}$, and $(iv)$ by coherence.
 \end{proof}

 \begin{corollary}
   When $\cat{C}$ is an coaffine rig category, $\mu_{A,B}$ is
   idempotent in $L(\cat{C})$.
 \end{corollary}
 \begin{proof}
   By the lemma above, $\mu_{I,I}$ mediates between $\mu_{A,B}$ and
   $\mu_{A,B} \circ \mu_{A,B}$.
 \end{proof}

The intuition here is that $\mu_{A,B}$ measures whether the state is
in the subspace $A$ or $B$, yielding a single bit of information. That
this bit can be mediated by $\mu_{I,I}$ means that we can either
perform the same measurement again (getting the same result), or just
copy the (classical) measurement outcome by measuring that.

 \begin{lemma}\label{lem:coproductdiagramcommutes}
   The diagram below commutes in any coaffine rig category.
 \[\begin{tikzcd}
  & {C \otimes (E \oplus E')} \\
  {C \otimes E} & {(C \otimes E) \oplus (C \otimes E')} & {C \otimes E'} \\
  A & {A \oplus B} & B
  \arrow[""{name=0, anchor=center, inner sep=0}, "{\mathrm{id} \otimes \amalg_1}", from=2-1, to=1-2]
  \arrow[""{name=1, anchor=center, inner sep=0}, "{\amalg_1}"', from=2-1, to=2-2]
  \arrow["{\delta_L^{-1}}"', from=2-2, to=1-2]
  \arrow[""{name=2, anchor=center, inner sep=0}, "{\mathrm{id} \otimes \amalg_2}"', from=2-3, to=1-2]
  \arrow[""{name=3, anchor=center, inner sep=0}, "{\amalg_2}", from=2-3, to=2-2]
  \arrow["f", from=3-1, to=2-1]
  \arrow[""{name=4, anchor=center, inner sep=0}, "{\amalg_1}"', from=3-1, to=3-2]
  \arrow["{f \oplus g}"', from=3-2, to=2-2]
  \arrow["g"', from=3-3, to=2-3]
  \arrow[""{name=5, anchor=center, inner sep=0}, "{\amalg_2}", from=3-3, to=3-2]
  \arrow["{(ii)}"{description}, draw=none, from=0, to=2-2]
  \arrow["{(ii)}"{description}, draw=none, from=2-2, to=2]
  \arrow["{(i)}"{description}, draw=none, from=4, to=1]
  \arrow["{(i)}"{description}, draw=none, from=5, to=3]
 \end{tikzcd}\]
 \end{lemma}
 \begin{proof}
 $(i)$ commutes by naturality of injections, $(ii)$ by \Cref{corr:coaffinediagrams}.
 \end{proof}

An object $A+B$ is a \emph{coproduct} of $A$ and $B$ when it has injections $A \to A+B \leftarrow B$, and pairs of morphisms $A \to C \leftarrow B$ factor through a unique morphisms $A+B \to C$ via the injections. This codifies classical control categorically~\cite{heunenvicary:cqm}.

\begin{proposition}\label{prop:coproducts}
  $\Split(\mathrm{LR(\cat{C})})$ has coproducts.
\end{proposition}
\begin{proof}
  For objects $(A,e)$ and $(B,e')$ of $\Split(\mathrm{LR(\cat{C})})$,
  we see that the diagrams
\[\begin{tikzcd}
  A & {A \oplus B} & B \\
  {A \otimes I} & {(A \oplus B) \otimes (I \oplus I)} & {B \otimes I}
  \arrow["{\amalg_1}", from=1-1, to=1-2]
  \arrow["{\rho^{-1}_\otimes}"', from=1-1, to=2-1]
  \arrow["{\mu_{A,B}}", from=1-2, to=2-2]
  \arrow["{\amalg_2}"', from=1-3, to=1-2]
  \arrow["{\rho^{-1}_\otimes}", from=1-3, to=2-3]
  \arrow["{\amalg_1 \otimes \amalg_1}"', from=2-1, to=2-2]
  \arrow["{\amalg_2 \otimes \amalg_2}", from=2-3, to=2-2]
\end{tikzcd}\]
commute in $\mathrm{R}(\cat{C})$ by
Corrolary~\ref{corr:coaffinediagrams}, so $\mu_{A,B} \circ
\amalg_1 = \amalg_1$ and $\mu_{A,B} \circ \amalg_2 = \amalg_2$ in
$\mathrm{LR}(\cat{C})$; call their equivalence classes in
$\mathrm{LR(\cat{C})}$ $i_1$ and $i_2$ respectively. Notice further
that $(e \oplus e') \circ \mu_{A,B} = \mu_{A,B} \circ (e \oplus e')$
by naturality of $\mu_{A,B}$, so $(e \oplus e') \circ \mu_{A,B}$ is
idempotent when $e$ and $e'$ are by
$$
(e \oplus e') \circ \mu_{A,B} \circ (e \oplus e') \circ \mu_{A,B}
= (e \oplus e') \circ (e \oplus e') \circ \mu_{A,B} \circ \mu_{A,B}
= (e \oplus e') \circ \mu_{A,B}
$$ Thus it follows by naturality of $\amalg_1$ and $\amalg_2$ that $(e
\oplus e') \circ \amalg_1$ and $(e \oplus e') \circ \amalg_2$ are
morphisms into $(A \oplus B, (e \oplus e') \circ \mu_{A,B})$ from
$(A,e)$ respectively $(B,e')$. Given morphisms $f : (A,e) \to (C,d)$
and $g : (B,e') \to (C,d)$ and choosing $[f,g]$ to be the equivalence
class of $\delta_L^{-1} \circ (f \oplus g) \circ \mu_{A,B}$ in
$\mathrm{LR(\cat{C})}$, we see that commutativity of the coproduct
triangle
\[\begin{tikzcd}
  & {(C,d)} \\
  \\
  {(A,e)} & {(A \oplus B, (e \oplus e') \circ \mu_{A,B})} & {(B,e')}
  \arrow["f", from=3-1, to=1-2]
  \arrow["{i_1}"', from=3-1, to=3-2]
  \arrow["{[f,g]}"', dashed, from=3-2, to=1-2]
  \arrow["g"', from=3-3, to=1-2]
  \arrow["{i_2}", from=3-3, to=3-2]
\end{tikzcd}\]
follows by Lemma~\ref{lem:coproductdiagramcommutes} and 
$\amalg_1 = \amalg_1 \circ \mu_{A,B}$ and $\amalg_2 = \amalg_2 \circ
\mu_{A,B}$ established above. Notice that $[f,g]$ respects the idempotents
since
\begin{align*}
  & d \circ [f,g] \circ (e
  \oplus e') \circ \mu_{A,B} \\
  & \quad = d \circ \delta_L^{-1} \circ (f \oplus g) \circ \mu_{A,B} \circ (e
  \oplus e') \circ \mu_{A,B} \\
  & \quad = d \circ \delta_L^{-1} \circ (f \oplus g) \circ (e
  \oplus e') \circ \mu_{A,B} \circ \mu_{A,B} \\
  & \quad = \delta_L^{-1} \circ (d \oplus d) \circ (f \oplus g) \circ (e
  \oplus e') \circ \mu_{A,B} \circ \mu_{A,B} \\
  & \quad = \delta_L^{-1} \circ ((d \circ f \circ e) \oplus (d \circ g \circ e'))
  \circ \mu_{A,B} \circ \mu_{A,B} \\
  & \quad = \delta_L^{-1} \circ (f \oplus g) \circ \mu_{A,B} \circ \mu_{A,B} \\
  & \quad = \delta_L^{-1} \circ (f \oplus g) \circ \mu_{A,B} \\
  & \quad = [f,g]
\end{align*}
in $\mathrm{LR(\cat{C})}$ by naturality of $\delta_L^{-1}$ and
$\mu_{A,B}$, idempotence of $\mu_{A,B}$ and the fact that $d \circ f
\circ e = f$ and $d \circ g \circ e'$ since they were assumed to be
morphisms $(A,e) \to (C,d)$ and $(B,e') \to (C,d)$ in
$\Split(\mathrm{LR(\cat{C})})$, so this is indeed a morphism between
the claimed objects in $\Split(\mathrm{LR(\cat{C})})$. This takes care
of existence.

To see that $[f,g]$ is also the unique such, suppose that $h$ is some
other morphism $(A \oplus B, (e \oplus e') \circ \mu_{A,B}) \to (C,d)$
making the coproduct triangle commute. By the fact that $h$ is a morphism
out of $(A \oplus B, (e \oplus e') \circ \mu_{A,B})$, we have that $h$
respects both $\mu_{A,B}$ and $(e \oplus e')$ individually since
\begin{align*}
h \circ \mu_{A,B} & = h \circ (e \oplus e') \circ \mu_{A,B} \circ
\mu_{A,B} = h \circ (e \oplus e') \circ \mu_{A,B} = h \\
h \circ (e \oplus e') & = h \circ (e \oplus e') \circ \mu_{A,B} \circ
(e \oplus e') = h \circ (e \oplus e') \circ (e \oplus e') \circ \mu_{A,B} \\
& = h \circ (e \oplus e') \circ \mu_{A,B} = h
\end{align*}
in $\mathrm{LR}(\cat{C})$. But then
\begin{align*}
  h & = h \circ \mu_{A,B} \\
  & = h \circ \mu_{A,B} \circ \mu_{A,B} \\
  & = h \circ \delta_R^{-1} \circ ((\id \otimes \amalg_1) \oplus (\id
  \otimes \amalg_2)) \circ \rho^{-1}_\otimes \oplus \rho^{-1}_\otimes
  \circ \mu_{A,B} \\
  & = h \circ \delta_L^{-1} \circ ((\amalg_1 \otimes \id) \oplus
  (\amalg_2 \otimes \id)) \circ (\rho^{-1}_\otimes \oplus
  \rho^{-1}_\otimes) \circ \mu_{A,B} \\
  & = \delta_L^{-1} \circ ((h \otimes \id) \oplus (h \otimes \id))
  \circ ((\amalg_1 \otimes \id) \oplus (\amalg_2 \otimes \id)) \circ
  (\rho^{-1}_\otimes \oplus \rho^{-1}_\otimes) \circ \mu_{A,B}\\
  & = \delta_L^{-1} \circ (((h \circ \amalg_1) \otimes \id) \oplus
  ((h \circ \amalg_2) \otimes \id)) \circ (\rho^{-1}_\otimes \oplus
  \rho^{-1}_\otimes) \circ \mu_{A,B} \\
  & = \delta_L^{-1} \circ ((f \otimes \id) \oplus
  (g \otimes \id)) \circ (\rho^{-1}_\otimes \oplus
  \rho^{-1}_\otimes) \circ \mu_{A,B} \\
  & = \delta_L^{-1} \circ (\alpha_\otimes \oplus \alpha_\otimes) \circ
  (\rho^{-1}_\otimes \oplus \rho^{-1}_\otimes) \circ (f \oplus g)
  \circ \mu_{A,B} \\
  & = (\id \otimes \rho^{-1}_\otimes) \circ \delta_L^{-1} \circ (f \oplus g)
  \circ \mu_{A,B} \\
  & = (\id \otimes \rho^{-1}_\otimes) \circ [f,g]
\end{align*}
in $\mathrm{R}(\cat{C})$, so $h$ and $[f,g]$ are in the same equivalence
class in $\mathrm{LR}(\cat{C})$.
\end{proof}

\section*{Completeness}\label{sec:completeness}

We are now ready to prove completeness of $\Split(\LR(\Pi_k))$ as in \Cref{thm:measurement} in the main article. We will do so by verifying that it is a model for a signature satisfying the algebraic theory of quantum computation. For the axioms we refer to~\cite{staton:quantum} and the proof of \Cref{prop:signature} and \Cref{thm:staton} below. 

\begin{proposition}\label{prop:signature}
  In $\Split(\LR(\cat{C}))$, with the exponential
  action of $I \oplus I$ on $\otimes$, the object $I$ is a structure
  for the signature of the algebraic theory of quantum
  computation~\cite{staton:quantum} (restricted to the unitaries of $\cat{C}$).
\end{proposition}
\begin{proof}
  Since $\Split(\cat{D})$ is symmetric monoidal when $\cat{D}$
  is, it follows that $\Split(\mathrm{LR(\cat{C})})$ is
  symmetric monoidal as well (since $LR(\cat{C})$ is symmetric
  monoidal), and so admits an exponential action with $I \oplus I$ by
  Lemma~\ref{lem:expact}. For the operations $\new \colon 0 \to 1$,
  $\measure \colon 1 \to (0,0)$, and $\apply_U \colon n \to n$
  we associate the morphisms
  \begin{align*}
    \sem{\new} & = \amalg_1 \colon I \to I \oplus I \\
    \sem{\measure} & = \mu_{I,I} \colon I \oplus I \to I+I \\
    \sem{\apply_U} & = F(U) \colon (I \oplus I)^{\otimes n} \to
    (I \oplus I)^{\otimes n}
  \end{align*}
  where $F$ is the symmetric monoidal functor $\cat{C} \to
  \Split(\mathrm{LR(\cat{C})})$. These operations match the
  signature since $m \bullet I = (I \oplus I)^{\otimes m}$ in the exponential action of $I \oplus I$ (up to canonical coherence isomorphism).
\end{proof}

\begin{definition}\cite{staton:quantum}
  A model of an algebraic theory with linear parameters is a structure
  for its signature such that for each axiom $x_1 : m_1, \ldots, x_k :
  m_k \mid a_1, \ldots, a_p \vdash t = u$ the interpretations
  $\sem{t},\sem{u} : p \bullet X \to (m_1 \bullet X) + \cdots + (m_k
  \bullet X)$ are equal.
\end{definition}

\begin{theorem}\label{thm:staton}
  In $\Split(\mathrm{LR(\cat{C})})$ with the exponential
  action of $I \oplus I$ on $\otimes$, the object $I$ forms a model of
  the algebraic theory of quantum computation~\cite{staton:quantum}.
\end{theorem}
\begin{proof}
  We verify the required axioms (A)-(L)~\cite{staton:quantum}. Write $2$ for the object $1+1$, write $\distribli$ and $\distribri$ for the coherence isomorphisms $A \otimes (B \oplus C) \to (A \otimes B) \oplus (A \otimes C)$ and $(A \oplus B) \otimes C \to (A \otimes C) \oplus (B \otimes C)$, and write $\lambda_A$ for the coherence isomorphism $I \otimes A \to A$.

  The first group of axioms concerns coaffine diagrams, which are seen to hold in $\cat{FinInj}$ by direct computation, and therefore hold in $\Split(\LR(\cat{C}))$ by \Cref{corr:coaffinediagrams}:
  \begin{enumerate}[itemsep=1mm]
    \item[(A)]
      $
        \sem{\measure} \circ X
        = \mu_{1,1} \circ \sigma^+ 
        = (\sigma^+ \times \sigma^+) \circ \mu_{1,1}
        = X \circ \sem{\measure}
      $; 

    \item[(D)]
      $\measure \circ \new = \new$; 

    \item[(J)]
      $(\distribri + \distribri) \circ \big((\measure \otimes \id) + (\measure \otimes \id)\big) \circ \distribli \circ (\id \otimes \measure)$ \\
      $= (\id +\sigma+\id) \circ (\distribli +\distribli) \circ \distribri \circ (\id \otimes \measure) \circ (\measure \otimes \id)$; 

    \item[(K)]
      $
        (\sem{\new} \otimes \id[2]) \circ (\id[2] \otimes \sem{\new})
        =
        (\id[2] \otimes \sem{\new}) \circ (\sem{\new} \otimes \id[2]) 
      $; 

    \item[(L)]
      $
        (\id[2] \otimes \measure) \circ (\new \otimes \id[2]) \circ \lambda_2^{-1} 
        = (\new \otimes \id[2]) \circ \lambda_2^{-1} \circ \measure
      $.
  \end{enumerate}
  For example, the natural transformation $\mu_{A,B} \colon A + B \to (A+B) \times (1+1)$ is given by $A \ni a \mapsto (a,\amalg_1(*))$ and $B \ni b \mapsto (b,\amalg_2(*))$ in $\cat{FinInj}$. 
  Write
  \[D(u,v) = (\distribri)^{-1} \circ (u \oplus v) \circ \distribri \colon 2^{\otimes n+1} \to 2^{\otimes n+1}
  \]
  for the map in $\LR(\cat{C})$ that models a computation on $n+1$ qubits, where the first qubit controls whether $u$ or $v$ is applied to the last $n$ qubits. The following two axioms are proven via \Cref{corr:coaffinediagrams} as above using naturality.
  \begin{enumerate}[itemsep=1mm]
    \item[(B)]
      $D(u,v) \circ (\measure \otimes \id[2^{\otimes n}])
       =(\measure \otimes \id[2^{\otimes n}]) \circ D(u,v)$;

    \item[(E)]
      $D(u,v) \circ (\new \otimes \id[2^{\otimes n}] \otimes \id[2^{\otimes n}]) = \new \otimes u$.
  \end{enumerate}
  \noindent
  The following group of axioms holds because the embedding $F \colon \cat{C} \to \Split(\LR(\cat{C}))$ is by construction a symmetric monoidal functor:
  \begin{enumerate}[itemsep=1mm]
    \item[(F)]
      $F(\sigma^\otimes_{2^{\otimes m},2^{\otimes n}}) = \sigma^\otimes_{2^{\otimes m},2^{\otimes n}}$;

    \item[(G)]
      $F(\id[2^{\otimes n}]) = \id[2^{\otimes n}]$; 
    
    \item[(H)]
      $F(v \circ u) = F(v) \circ F(u)$;
    
    \item[(I)]
      $F(u \otimes v) = F(u) \otimes F(v)$; 
  \end{enumerate}
  \noindent
  The final axiom follows from the nature of the construction of $\Split(\LR(\cat{C}))$:
  \begin{enumerate}[itemsep=1mm]
    \item[(C)]
      It follows from the fact that $I$ is terminal in $\Split(\LR(\cat{C}))$ that $d \circ u = d$ for any $u \colon 2^{\otimes n} \to 2^{\otimes n}$ and $d \colon 2^{\otimes n} \to I$ in $\cat{C}$, so in particular for the discarding map $d \colon 2^{\otimes n} \to 1$ in $\Split(\LR(\cat{C}))$ lifted from $\lambda^{-1} \colon 2^{\otimes n} \to I \otimes 2^{\otimes n}$ in $\cat{C}$.\qedhere
  \end{enumerate}
\end{proof}

Let $\cat{QChannel}$ denote the category of
finite-dimensional $C^*$-algebras and quantum channels (completely positive trace-preserving linear maps). The
following corollary to Theorem~\ref{thm:staton} establishes the first
part of \Cref{thm:measurement} in the main article.

\begin{corollary}
  The induced functor $\sem{-} \colon \Split(\LR(\Pi_k)) \to
  \cat{QChannel}$ is faithful.
\end{corollary}
\begin{proof}
  It was established in \cite{staton:quantum} that every equality of
  quantum channels can be established using the discrete set of axioms
  shown to hold in $\Split(\LR(\Pi_k))$ by Theorem~\ref{thm:staton}.
\end{proof}

Finally, we conclude that $\Split(\LR(\Pi_k))$ is complete for unitaries,
establishing the last part of \Cref{thm:measurement} in the main article.

\begin{theorem}\label{completeness}
  Parallel morphisms $f$ and $g$ in $\Pi_k$ are equal when included
  into $\Split(\LR(\Pi_k))$ if and only if $f \approx_k s \cdot g$
  for some scalar $s$.
\end{theorem}
\begin{proof}
  Suppose $f \approx_k s \cdot g$ in $\Pi_k$, i.e., there exist $f'$
  and $g'$ such that $f \oplus f' = (s \cdot g) \oplus g'$. Note that
  $f \oplus f'$ and $(s \cdot g) \oplus g'$ are precisely
  interpretations of the circuits
  \begin{center}
    \includegraphics{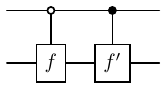} \qquad
    \includegraphics{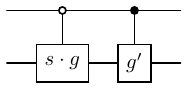}
  \end{center}
  But then $f = g$ in $\Split(\LR(\Pi_k))$ follows by
  \begin{center}
    \includegraphics{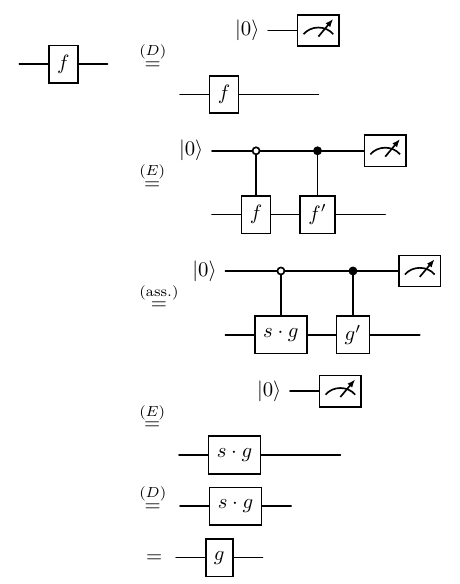}\qedhere
  \end{center}
  where the last step follows as all scalars are identified in
  $\Split(\LR(\Pi_k))$ (since $I$ is terminal), so $s \cdot g = \id
  \cdot g = g$.

  In the other direction, when $f$ and $g$ are unitaries in $\Pi_k$
  such that $f = g$ when included in $\Split(\LR(\Pi_k))$, it follows
  by functoriality of $\sem{-} \colon \Split(\LR(\Pi_k)) \to
  \cat{QChannel}$ that $\sem{f} = \sem{g}$, and since
  $f$ and $g$ were included from $\Pi_k$ these quantum channels in
  fact just adjoining by the underlying unitaries of $f$ and $g$.
  Since adjoining by the underlying unitaries of $f$ and $g$ give
  equal quantum channels, it follows by unitary freedom~\cite[Theorem
    2.6]{nielsenchuang} of their Choi-states that their underlying
  unitaries are equal up to a global phase.  Applying
  Theorem~\ref{thm:completeness} then establishes $f \approx_k s
  \cdot g$ in $\Pi_k$.
\end{proof}

\section*{Asymptotic scaling with $k$}

We study how the $\Pi_k$ term size needed to approximate a target unitary depends on the precision level~$k$. Given an overall approximation error~$\varepsilon$ in operator norm, how does the minimal $\Pi_k$ term size scale with $k$ for different classes of unitaries? Two representative cases illustrate the range of behaviours: generic unitaries and the quantum Fourier transform.

\paragraph{Generic unitaries.}
Consider circuits made of single-qubit rotations and CNOT gates. A rotation $R_a(\phi)$ can be exactly represented by a term of $\Pi_k$ if and only if the entries of $R_a(\phi)$ lie in $\DD[\zetak]$~\cite{PhysRevLett.110.190502}; all other rotations must be approximated. For such generic rotations, any $\varepsilon$-approximation requires $O(\log(1/\varepsilon))$ primitive terms and approximations matching this lower bound can be constructed efficiently~\cite{rossselinger:synthesis}. If a constant fraction of the rotations in a circuit are generic, the overall size of the corresponding $\Pi_k$ term therefore scales as $O(\log(1/\varepsilon))$ for any fixed $k$, with larger $k$ potentially improving only the hidden constant factor. Thus, no finite $k$ can remove the logarithmic dependence on $\varepsilon$, only improve its constant.

\paragraph{The quantum Fourier transform.}
The $n$-qubit quantum Fourier transform ($\text{QFT}_n$) contains $n$ Hadamard gates and $\binom{n}{2}$ controlled-phase gates. The angles appearing in the phase gates are of the form $2\pi / 2^d$, with $d$ ranging from $2$ to $n$. At precision level $k$, all rotations with $d\leq k$ are native to $\Pi_k$ and therefore require no approximation. Hence, when $n\leq k$, no approximations are needed at all, and $O(n^2)$ $\Pi_k$-terms suffice to represent $\text{QFT}_n$ exactly. When $n > k$,
\[
  \binom{(n-k)+1}{2}
\]
controlled-phase gates cannot be exactly represented in $\Pi_k$ and must therefore be approximated. For example, a single gate requires approximation when $k=n-1$, three gates require approximations when $k=n-2$, and so on. Overall, assuming that the total error budget $\varepsilon$ is split evenly among all gates requiring approximation, standard methods~\cite{rossselinger:synthesis} will therefore produce $\Pi_k$ terms of size $O ((n^2 - (n-k)^2) + (n-k)^2\log_2((n-k)^2/\varepsilon)))$. Hence, in the case of the quantum Fourier transform, increasing $k$ can remove the dependence on $\varepsilon$ altogether: As $k$ increases, the residual cost drops until $k=n$, after which all rotations are native and the size plateaus at $O(n^2)$.

\end{document}